\documentclass[a4paper,11pt]{amsart}
\usepackage{amssymb,dsfont}
\usepackage{cmap}
\usepackage{hyperxmp}
\usepackage{graphicx}
\usepackage[pdfdisplaydoctitle=true,
      colorlinks=true,
      urlcolor=blue,
      citecolor=blue,
      linkcolor=blue,
      pdfstartview=FitH,
      pdfpagemode= UseNone,
      bookmarksnumbered=true]{hyperref} 
\usepackage[all]{xy}
\usepackage{cleveref}
\newtheorem{thm}{Theorem}[section]
\newtheorem{cor}[thm]{Corollary}
\newtheorem{lem}[thm]{Lemma}
\newtheorem{prop}[thm]{Proposition}
\theoremstyle{definition}
\newtheorem{defn}[thm]{Definition}
\theoremstyle{remark}
\newtheorem{rem}[thm]{Remark}
\newtheorem{ex}[thm]{Example}
\numberwithin{equation}{section}
\newcommand{\CC}{\mathbb{C}}                
\newcommand{\RR}{\mathbb{R}}                
\newcommand{\ZZ}{\mathbb{Z}}                

\newcommand{\Ela}{\mathbb{E}\mathrm{la}}    
\newcommand{\HH}{\mathbb{H}}                
\newcommand{\TT}{\mathbb{T}}                
\newcommand{\Sym}{\mathbb{S}}               

\newcommand{\Sn}[1]{\mathrm{S}^{#1}}        
\newcommand{\SRn}[1]{\mathrm{S}^{#1}_{\mathbb{\RR}}(\CC^{2})}       
\newcommand{\Hn}[1]{\mathrm{H}^{#1}}        

\newcommand{\SL}{\mathrm{SL}}               
\newcommand{\SU}{\mathrm{SU}}               
\newcommand{\OO}{\mathrm{O}}                
\newcommand{\SO}{\mathrm{SO}}               
\newcommand{\octa}{\mathbb{O}}              
\newcommand{\ico}{\mathbb{I}}               
\newcommand{\tetra}{\mathbb{T}}             
\newcommand{\DD}{\mathbb{D}}                
\newcommand{\perm}{\mathfrak{S}}            
\newcommand{\triv}{\mathds{1}}		        

\newcommand{\be}{\pmb{e}}                   
\newcommand{\nn}{\pmb{n}}                   
\newcommand{\vv}{\pmb{v}}                   
\newcommand{\ww}{\pmb{w}}                   
\newcommand{\xx}{\pmb{x}}                   
\newcommand{\bxi}{\pmb{\xi}}                

\newcommand{\ff}{\mathbf{f}}                
\newcommand{\bg}{\mathbf{g}}                
\newcommand{\bw}{\mathbf{w}}                

\newcommand{\Idd}{\mathbf{1}}               
\newcommand{\qq}{\mathrm{q}}                
\newcommand{\rp}{\mathrm{p}}                
\newcommand{\rh}{\mathrm{h}}                

\newcommand{\ba}{\mathbf{a}}
\newcommand{\bb}{\mathbf{b}}
\newcommand{\bd}{\mathbf{d}}
\newcommand{\bh}{\mathbf{h}}

\newcommand{\bs}{\mathbf{s}}
\newcommand{\bv}{\mathbf{v}}
\newcommand{\blambda}{\pmb \lambda}

\newcommand{\bH}{\mathbf{H}}                
\newcommand{\bT}{\mathbf{T}}                

\DeclareMathOperator{\Ad}{Ad}
\DeclareMathOperator{\tr}{tr}

\newcommand{\norm}[1]{\lVert#1\rVert}       
\newcommand{\abs}[1]{\lvert#1\rvert}        
\newcommand{\set}[1]{\left\{#1\right\}}     

\newcommand{\tq}[1]{{{\mathbf{#1}}}}
\newcommand{\td}[1]{{\mathbf{#1}}}
\renewcommand{\vec}{\pmb}

\hypersetup{
  pdfauthor={M. Olive, B. Kolev, B. Desmorat and R. Desmorat}, 
  pdftitle={Harmonic factorization and reconstruction of the elasticity tensor}, 
  pdfsubject={MSC 2010: 74E10 (15A72, 74B05)}, 
  pdfkeywords={Anisotropy; Sylvester theorem; Harmonic factorization; Harmonic product; Tensorial reconstruction; Covariant tensors}, 
  pdflang=en, 
  }
\begin{document}

\title[Harmonic factorization]{Harmonic factorization and reconstruction of the elasticity tensor}%

\author{M. Olive}
\address[Marc Olive]{LMT  (ENS Cachan, CNRS, UMR 8535, Universit\'{e} Paris Saclay), F-94235 Cachan Cedex, France}
\email{marc.olive@math.cnrs.fr}

\author{B. Kolev}
\address[Boris Kolev]{Aix Marseille Universit\'{e}, CNRS, Centrale Marseille, I2M, UMR 7373, F-13453 Marseille, France}
\email{boris.kolev@math.cnrs.fr}

\author{B. Desmorat}
\address[Boris Desmorat]{Sorbonne Universit\'{e}, UMPC Univ Paris 06, CNRS, UMR 7190, Institut d'Alembert, F-75252 Paris Cedex 05, France}
\email{boris.desmorat@upmc.fr}
\address[Boris Desmorat]{Univ Paris Sud 11, F-91405 Orsay, France}

\author{R. Desmorat}
\address[Rodrigue Desmorat]{LMT (ENS Cachan, CNRS, UMR 8535, Universit\'{e} Paris Saclay), F-94235 Cachan Cedex, France}
\email{desmorat@lmt.ens-cachan.fr}


\date{\today}%
\begin{abstract}
  In this paper, we study anisotropic Hooke's tensor: we propose a factorization of its fourth-order harmonic part into second-order tensors. We obtain moreover explicit equivariant reconstruction formulas, using second-order covariants, for transverse isotropic and orthotropic fourth-order harmonic tensors, and for trigonal and tetragonal fourth-order harmonic tensors up to a cubic fourth order covariant remainder.
\end{abstract}

\maketitle


\begin{scriptsize}
  \setcounter{tocdepth}{2}
  \tableofcontents
\end{scriptsize}

\section{Introduction}\label{sec:Intro}

Interest in coordinate-free representations formulas for linear anisotropic elastic materials has been an active research area, starting by the formulation of elastic energy functionals in the framework of finite strains~\cite{TN1965,Smi1965,Spe1984,Hol2000} and having a key role in the classification of linear elastic or piezoelectric materials~\cite{FV1996,GW2002}. It has also been of main importance in the Continuum Mechanics representation of cracked/damaged media \cite{Cha1979,LO1980,CS1982,Ona1984}.

One underlying difficulty is that several Elasticity tensors may represent the same linear anisotropic elastic material in different orientations. More precisely, any change of orientation of a material specified by some rotation $g\in \SO(3)$ defines a new Elasticity tensor $\overline{\tq{E}}$ deduced from the previous one by some \emph{group action}
\begin{equation*}
  \tq{E}\mapsto \overline{\tq{E}} = g\star \tq{E},\qquad \overline{E}_{ijkl} = g_{ip}g_{jq}g_{kr}g_{ls} E_{pqrs},
\end{equation*}
where Einstein convention on repeated indices is used. As the material is rotated, the tensor $\tq{E}$ moves on its \emph{orbit} in the space $\Ela$ of Elasticity tensors. Thus, a linear elastic material is represented by the \emph{orbit} of an Elasticity tensor under this group action, and any intrinsic parameter is necessary an \emph{invariant} of this group action.

In a seminal paper, Boehler--Kirilov--Onat~\cite{BKO1994} emphasized the fundamental role played by polynomial \emph{invariants} of the Elasticity tensor and this has been used by Auffray--Kolev--Petitot~\cite{AKP2014} to classify the orbits of the elasticity tensor. The old problem of finding a basis of polynomial invariants for the fourth order Elasticity tensor was finally solved by Olive in 2014~\cite{Oli2014} after several attempts~\cite{BKO1994,BH1995,SB1997,Xia1997,Ost1998}. A \emph{minimal integrity basis of $297$ invariants} was definitively obtained in~\cite{OKA2017}. Note that in 2D, this integrity basis is composed by only $6$ invariants~\cite{Ver1979,Via1997}, which shows the incredible complexity of 3D linear elasticity compared to 2D. Moreover, in 3D, the problem cannot be reduced to the question of finding an integrity basis for second-order tensor-valued functions.

Two main tools have been used to solve this extremely difficult computational problem.

The first one is the so-called \emph{harmonic decomposition} which corresponds to the splitting of the Elasticity tensor into \emph{irreducible pieces}~\cite{Sch1951,Spe1970} and which was first achieved by Backus in 1970~\cite{Bac1970}. Such a decomposition is useful to compute the \emph{material symmetry classes}~\cite{FV1996,FBG1996,GW2002}. As mentioned, it has been used also in the study of effective elastic properties of cracked media~\cite{LO1980,Kan1984,Ona1984,BHL1995,CW2010} and -- without explicit reference to it -- in the homogenization techniques of laminated composites \cite{VV2001,VP2006,MVV2012}. Note that in the later case, the polar decomposition method for 2D media has been used \cite{Ver1979,Van2005}, while its link with 2D harmonic decomposition has been shown in~\cite{FV2014,DD2015}.

The second one is a purely mathematical tool, \emph{binary forms}, which was the cornerstone of Invariant Theory in the nineteenth century~\cite{Cay1861,Gor1868,Gor1875,Gor1987,GY2010} and is connected to spinors. This tool has been brought first to the knowledge of the mechanical community by Backus~\cite{Bac1970}, and then by Boehler--Kirilov--Onat~\cite{BKO1994} but its deep power did not seem to have been widely considered so far. It happens to be extremely useful to solve problems in tensor analysis for higher order tensors in 3D. For instance, this tool was used by Auffray--Olive~\cite{OA2014} to produce a minimal integrity basis for a traceless and totally symmetric third order tensor, which appears in piezoelectricity~\cite{Yan2009} and second-gradient strain elasticity theory~\cite{Min1965}. Note, furthermore, that applications of binary forms are not just bounded to questions in continuum mechanics but are also related to other fields such as quantum computation~\cite{Luq2007} and cryptography~\cite{LR2012}.

The harmonic decomposition of the Elasticity tensor produces a quintuple
\begin{equation*}
  (\alpha,\beta,\ba',\bb',\bH)
\end{equation*}
where $\alpha,\beta,$ are scalar invariants, $\ba',\bb'$ are second order traceless tensors and $\bH$ is a \emph{fourth order harmonic tensor}. The scalar and second order tensor components are linearly related to the Voigt and the dilatation tensors, the two independent traces of the Elasticity tensor. The irreducible fourth order component remains problematic in the study the Elasticity tensor.

In the present work, we try to go further in the representation of the harmonic fourth order component $\bH$. We propose to introduce a secondary decomposition of this tensor, which was first discovered by Sylvester~\cite{Syl1909} and now known as \emph{Maxwell multipoles} (see also~\cite{Bac1970,Bae1998}), through a binary operation called \emph{harmonic product} (a commutative and associative product directly inherited from the product of binary forms). This decomposition aims ideally to reduce $\bH$ to a family of vectors, the Maxwell multipoles, or, as detailed next, to reduce $\bH$ to a family of second order tensors. In 2D, the factorization of the harmonic component $\bH$ by means of second order tensors was already established in~\cite{DD2015}, it results from the decomposition of the Elasticity tensor by Verchery's polar method \cite{Ver1979}. Besides, this factorization has the good taste to be \emph{equivariant}, which means that it commutes with the action of the rotation group.

However, the solution in 3D of such a decomposition -- \emph{i.e.} the four Maxwell multipoles or in an equivalent manner the two second order tensors that factorize $\bH$ -- is far from being unique and constructive. It is thus of poor value in practice. Moreover, there are no global \emph{equivariant sections} for such mappings. We concentrate therefore on producing \emph{equivariant explicit solutions} of the problem that we call a \emph{reconstruction} but which are limited to specific symmetry classes. This process uses the concept of \emph{covariants}, which generalizes the idea of invariants but which are tensors (rather than scalars) depending in an equivariant manner of the original tensor. For obvious reasons, no such reconstruction can be globally defined. As it is well known \cite{Liu1982}, only triclinic, monoclinic, orthotropic, transversely isotropic and isotropic classes have a chance to be reconstructed using second-order covariants.

In this paper, we obtain for the first time explicit reconstruction formulas in the orthotropic and the transversely isotropic cases, using \emph{polynomial} second-order harmonic covariants and rational invariants. Moreover, to overpass the geometrical constraints in the trigonal and in the tetragonal cases, we establish for these symmetry classes reconstruction formulas by means of second order covariants up to a single \emph{cubic fourth order covariant remainder}.

The considered reconstruction problem is closely related to the so-called \emph{isotropic extension} of anisotropic constitutive functions via structural tensors developed by Boehler~\cite{Boe1979} and Liu~\cite{Liu1982} independently (see also~\cite{Boe1987,Man2016}). In the case of a \emph{linear constitutive law}, an isotropic extension is just an \emph{equivariant reconstruction} \emph{limited to a given symmetry class} of the constitutive tensor. Furthermore, our approach is more constructive since we are able to give \emph{explicit equivariant formulas} in which the ``structural tensors'' are polynomial covariant tensors of the constitutive tensor.

\subsection*{Organization of the paper}

In Section~\ref{sec:tensorial-representations}, we recall basic materials on tensorial representations of the orthogonal group, in particular, the link between totally symmetric tensor and homogeneous polynomials and the harmonic decomposition. In Section~\ref{sec:harmonic-factorization}, we introduce the \emph{harmonic product} between harmonic polynomials and formulate an harmonic factorization theorem which proof is postponed in Appendix~\ref{sec:Annexe_Sylvester-theorem-proof}. The general reconstruction problem is formulated in Section~\ref{sec:reconstruction}, where a geometric obstruction for an equivariant reconstruction of fourth-order tensors, by means of second order polynomial covariants is explained. Explicit reconstructions formulas, using rational invariants and second order polynomial covariants are obtained in Section~\ref{sec:second-order-covariants} for transversely isotropic and orthotropic tensors. In Section~\ref{sec:other-results}, we propose similarly equivariant reconstructions for tetragonal and trigonal fourth order harmonic tensors up to a cubic covariant remainder.

\subsection*{Notations}

The following spaces are involved:
\begin{itemize}
  \item $\Ela$: the space of \emph{Elasticity tensors};
  \item $\TT^{n}(\RR^{3})$: the space of $n$-th order tensors on $\RR^{3}$;
  \item $\Sym^{n}(\RR^{3})$: the space of $n$-th order \emph{totally symmetric tensors} on $\RR^{3}$;
  \item $\HH^{n}(\RR^{3})$: the space of $n$-th order \emph{harmonic tensors} on $\RR^{3}$;
  \item $\Sn{n}(\RR^{3})$: the space of homogeneous polynomials of degree $n$ on $\RR^{3}$;
  \item $\Hn{n}(\RR^{3})$: the space of \emph{harmonic polynomials} of degree $n$ on $\RR^{3}$;
  \item $\Sn{n}(\CC^{2})$: the space of binary forms of degree $n$ on $\CC^{2}$;
  \item $\SRn{2n}$: the space of binary forms of degree $2n$ which correspond to real harmonic tensors of degree $n$.
\end{itemize}
In addition, we adopt the following conventions:
\begin{itemize}
  \item $x, y, z, x_i$: coordinates on $\RR^{3}$;
  \item $u,v$: coordinates on $\CC^{2}$;
  \item $\vv$, $\ww$, $\xx$,
        $\vec{n}$, $\vec{e}_i$: vectors of $\RR^{3}$;
  \item $\bxi$: a vector in $\CC^{2}$;
  \item $\ba, \bb$, $\bd_k$, $\td s$, $\bh$, $\bh_{k}$: second-order tensors;
  \item $\bd$: second-order dilatation tensor;
  \item $\bv$: second-order Voigt tensor;
  \item $\td 1$: second order unit tensor;
  \item $\tq E$: elasticity tensor;
  \item $\tq I$: fourth-order unit tensor;
  \item $\bT$, $\tq S$: generic tensors;
  \item $\bH$: fourth-order harmonic tensor;
  \item $\rp$: polynomial on $\RR^{3}$;
  \item $\rh$: harmonic polynomial on $\RR^{3}$;
  \item $\ff, \bg, \bw$: binary forms;
  \item $g$: element of $\SO(3)$;
  \item $\cdot$\,: the scalar product;
  \item $\otimes$: the tensor product;
  \item $\odot$, $\otimes_{(4)}$: the symmetric tensor product;
  \item $\ast$: the harmonic product;
  \item $\star$: the action;
  \item $\norm{\bT} = \sqrt{\bT\cdot \bT}$: the Euclidean norm;
  \item $(.)^t$: the transpose;
  \item $(.)'$: the deviatoric part;
  \item $(.)_{0}$: the harmonic projection;
  \item $\overline {(.)}, \overline \lambda $: the complex conjugate.
\end{itemize}

\section{Tensorial representations of the rotation group}
\label{sec:tensorial-representations}

In this section, we recall classical facts about tensorial representations of the rotation group $\SO(3)$. More details on the subject can be found in~\cite{Ste1994} or~\cite{GSS1988}. We consider thus tensors on the Euclidean space $\RR^{3}$ and, thanks to the Euclidean product, we do not have to distinguish between \emph{upper} and \emph{lower} indices. Therefore, an $n$-th order tensor may always be considered as a $n$-linear mapping
\begin{equation*}
  \bT: \RR^{3} \times \dotsb \times \RR^{3} \to \RR, \qquad (\xx_{1},\dotsc,\xx_n) \mapsto \bT(\xx_{1},\dotsc,\xx_n).
\end{equation*}
Let $\TT^{n}(\RR^{3})$ be the space of $n$-th order tensors. The rotation group $\SO(3)$ acts on $\TT^{n}(\RR^{3})$, by the rule
\begin{equation*}
  (g \star \bT)(\xx_{1},\dotsc,\xx_n) := \bT(g^{-1} \cdot \xx_{1},\dotsc,g^{-1} \cdot \xx_{n}),
\end{equation*}
where $\bT \in \TT^{n}(\RR^{3})$ and $g \in \SO(3)$.

\begin{rem}
  The reason why, in the definition of the action of the rotation group on tensors, we set $g^{-1}$ rather than $g$ on the right hand-side of the defining formula is to get a \emph{left action}, meaning that: $(g_{1}g_{2}) \star \bT = g_{1} \star (g_{2} \star \bT)$, rather than a \emph{right action}: $(g_{1}g_{2}) \star \bT = g_{2} \star (g_{1} \star \bT)$.
\end{rem}

Usually, we are not interested in the whole tensor space $\TT^{n}(\RR^{3})$ but rather in a subspace $\mathbb{V}$ defined by some particular \emph{index symmetries} and \emph{stable} (or invariant) under the action of the rotation group $\SO(3)$, which means that
\begin{equation*}
  g\star \bT\in \mathbb{V}, \qquad \forall \bT\in \mathbb{V}, \quad \forall g\in \SO(3).
\end{equation*}

\begin{ex}\label{ex:symmetric-tensors}
  The permutation group $\perm_{n}$ acts on $\TT^{n}(\RR^{3})$ by the rule
  \begin{equation*}
    (\sigma \star \bT)(\xx_{1},\dotsc,\xx_n) := \bT(\xx_{\sigma^{-1}(1)}, \dotsc, \xx_{\sigma^{-1}(n)}), \qquad \sigma \in \perm_{n},
  \end{equation*}
  and this action commutes with the action of $\SO(3)$. The space of \emph{totally symmetric tensors} of order $n$, denoted by $\Sym^{n}(\RR^{3})$, is the space of tensors $\bT$ which are invariant under this action, in other words, such that
  \begin{equation*}
    (\sigma \star \bT)(\xx_{1},\dotsc,\xx_n) = \bT(\xx_{1},\dotsc,\xx_n), \qquad \forall \sigma \in \perm_{n}.
  \end{equation*}
\end{ex}

\begin{ex}\label{ex:antisymmetric-tensors}
  The space of \emph{alternate tensors} of order $n$, denoted by $\Lambda^{n}(\RR^{3})$, is the space of tensors $\bT$ such that
  \begin{equation*}
    (\sigma \star \bT)(\xx_{1},\dotsc,\xx_n) = \epsilon(\sigma) \bT(\xx_{1},\dotsc,\xx_n), \qquad \forall \sigma \in \perm_{n},
  \end{equation*}
  where $\epsilon(\sigma)$ is the signature of the permutation $\sigma$.
\end{ex}

\subsection{Harmonic tensors}

The representation $\mathbb{V}$ is \emph{irreducible} if the only stable subspaces are $\set{0}$ and $\mathbb{V}$. The irreducible representations of the rotation group $\SO(3)$ are known (up to isomorphism)~\cite{GSS1988}.
One explicit model for
these irreducible representations is given by the so-called \emph{harmonic tensor spaces} which are described below.

Contracting two indices $i,j$ on a totally symmetric tensor $\bT$ does not depend on the particular choice of the pair $i,j$. Thus, we can refer to this contraction without any reference to a particular choice of indices. We will denote this contraction as $\tr \bT$, which is a totally symmetric tensor of order $n-2$ and is called the \emph{trace} of $\bT$. Iterating the process leads to
\begin{equation*}
  \tr^{k} \bT = \tr(\tr(\dotsb (\tr \bT)))
\end{equation*}
which is a totally symmetric tensor of order $n-2k$.

\begin{defn}
  A \emph{harmonic tensor} of order $n$ is a totally symmetric tensor $\bT \in \Sym^{n}(\RR^{3})$ such that $\tr \bT = 0$. The space of harmonic tensors of order $n$ will be denoted by $\HH^{n}(\RR^{3})$. It is a sub-vector space of $\Sym^{n}(\RR^{3})$ of dimension $2n+1$.
\end{defn}

The sub-space $\HH^{n}(\RR^{3})$ of $\Sym^{n}(\RR^{3})$ is invariant under the action of $\SO(3)$ and is irreducible (see for instance~\cite{GSS1988}). Moreover, every finite dimensional irreducible representation of $\SO(3)$ is isomorphic to some $\HH^{n}(\RR^{3})$
and every $\SO(3)$-representation $\mathbb{V}$ splits into a direct sum of \emph{harmonic tensor spaces} $\HH^{n}(\RR^{3})$ (see~\cite{BtD1995,Ste1994}). Such a direct sum is called the \emph{harmonic decomposition} of $\mathbb{V}$.

\begin{ex}
  The harmonic decomposition of a second order symmetric tensor $\bs$ corresponds to its decomposition $\bs' + \frac{1}{3}(\tr \bs)\, \Idd$ into its deviatoric and spheric parts
  \begin{equation*}
    \Sym^{2}(\RR^{3}) \,\simeq\, \HH^{2}(\RR^{3})\oplus \HH^{0}(\RR^{3}).
  \end{equation*}
  In this formula, the symbol $\simeq$ means that there exists a linear isomorphism between the two vector spaces $\Sym^{2}(\RR^{3})$ and $\HH^{2}(\RR^{3})\oplus \HH^{0}(\RR^{3})$ which, moreover, commutes with the action of the rotation group on both sides (and is thus called an \emph{equivariant isomorphism}).
\end{ex}

There is a well-known correspondence between totally symmetric tensors of order $n$ and homogeneous polynomial of degree $n$ on $\RR^{3}$, which extends the well-known correspondence between a symmetric bilinear form and a quadratic form via polarization. Indeed, to each symmetric tensor $\bT \in \Sym^{n}(\RR^{3})$ corresponds a homogeneous polynomial of degree $n$ given by
\begin{equation*}\label{eq:isomorphism-Sn-Pn}
  \rp(\xx) := \bT(\xx,\dotsc,\xx), \qquad \xx := (x, y, z)\in \RR^{3}.
\end{equation*}

This correspondence defines a linear isomorphism $\phi$ between the tensor space $\Sym^{n}(\RR^{3})$ and the polynomial space $\Sn{n}(\RR^{3})$ of homogeneous polynomials of degree $n$ on $\RR^{3}$. 
The linear isomorphism $\phi$ between these spaces commutes with the action of the rotation group:
\begin{equation*}
  \phi(g \star \bT) = g \star \phi(\bT), \qquad g \in \SO(3),
\end{equation*}
and is thus an \emph{equivariant isomorphism}.

\begin{rem}
  The rotation group $\SO(3)$ acts on the polynomial space $\Sn{n}(\RR^{3})$, by the rule
  \begin{equation*}
    (g \star \rp)(\xx) := \rp(g^{-1} \cdot \xx), \qquad g \in \SO(3).
  \end{equation*}
\end{rem}

The inverse $\bT = \phi^{-1}(\rp)$ can be recovered by \emph{polarization}. More precisely, the expression $\rp(t_{1}\xx_{1} + \dotsb + t_{n}\xx_{n})$ is a homogeneous polynomial in the variables $t_{1},\dotsc,t_n$ and we get
\begin{equation*}
  \bT(\xx_{1},\dotsc,\xx_{n}) = \frac{1}{n!} \left.\frac{\partial^{n}}{\partial t_{1} \dotsb \partial t_{n}}\right|_{t_{1} = \dotsb = t_{n} = 0}\rp(t_{1}\xx_{1} + \dotsb + t_{n}\xx_{n}).
\end{equation*}

\begin{rem}
  Note that the \emph{Laplacian operator} of $\rp = \phi(\bT)$ writes as
  \begin{equation*}
    \triangle \phi(\bT) = n(n-1) \phi(\tr \bT).
  \end{equation*}
  Thus, \emph{totally symmetric tensors with vanishing trace} correspond \textit{via} $\phi$ to \emph{harmonic polynomials} (polynomials with vanishing Laplacian). This justifies the denomination of \emph{harmonic tensors} for elements of $\HH^{n}(\RR^{3})$. More generally, for any non-negative integer $k$, we get
  \begin{equation}\label{eq:Lapl-k}
    \triangle^{k} \phi(\bT) = \frac{n!}{(n-2k)!} \phi(\tr^{k} \bT).
  \end{equation}
\end{rem}

The equivariant isomorphism $\phi$ sends $\HH^{n}(\RR^{3})$ to $\Hn{n}(\RR^{3})$, the space of homogeneous harmonic polynomials of degree $n$.
Thus, the spaces $\Hn{n}(\RR^{3})$ provide a second model for \emph{irreducible representations} of the rotation group $\SO(3)$. A third model of these irreducible representations is provided by the spaces of \emph{binary forms} $\SRn{2n}$, whose construction is detailed in Appendix~\ref{sec:Annexe_binary-forms}.

\subsection{Harmonic decomposition of totally symmetric tensors}
\label{subsec:totally-symmetric-tensors}

The harmonic decomposition of a homogeneous polynomial of degree $n$ is described by the following proposition.

\begin{prop}\label{prop:symmetric-harmonic-decomposition}
  Every homogeneous polynomial $\rp \in \Sn{n}(\RR^{3})$ can be decomposed uniquely as
  \begin{equation}\label{eq:symmetric-harmonic-decomposition}
    \rp = \rh_{0} + \qq\,\rh_{1} + \dotsb + \qq^{r}\rh_{r},
  \end{equation}
  where
  $\qq(x,y,z) = x^{2} + y^{2} + z^{2}$, $r=[n/2]$ -- with $[\cdot]$ integer part -- and $\rh_{k}$ is a harmonic polynomial of degree $n-2k$.
\end{prop}

The harmonic polynomial $(\rp)_{0}:=\rh_{0}$ is the orthogonal projection of $\rp$ onto $\Hn{n}(\RR^{3})$ (for the induced Euclidean structure of $\RR^{3}$). Its explicit expression can be obtained recursively as follows (see~\cite{OKA2017}). Let first define
\begin{equation}\label{eq:hr-poly}
  \rh_{r} = \left\{
  \begin{array}{ll}
    \displaystyle{\frac{1}{(2r+1)!}}\triangle^{r} \rp,              & \hbox{if $n$ is even;} \\
    \displaystyle{\frac{3!\times (r+1)}{(2r+3)!}}\triangle^{r} \rp, & \hbox{if $n$ is odd.}
  \end{array}
  \right.
\end{equation}
Then $\rh_{k}$ is obtained inductively by the formula
\begin{equation}\label{eq:hk-poly}
  \rh_{k} = \mu(k) \triangle^{k}\left(\rp - \sum_{j=k+1}^{r} \qq^{j} \rh_{j}\right)
\end{equation}
with
\begin{equation*}
  \mu(k) := \frac{(2n-4k+1)!(n-k)!}{(2n-2k+1)!k!(n-2k)!},
\end{equation*}
which leads to $\rh_{0}$ after $r$ iterations.

Introducing the symmetric tensor product
\begin{equation*}
  \bT_{1} \odot \bT_{2} := \frac{1}{(n_{1} + n_{2})!} \sum_{\sigma \in \perm_{n_{1} + n_{2}}} \sigma \star \left(\bT_{1} \otimes \bT_{2}\right),
\end{equation*}
for $\bT_{1} \in \Sym^{n_{1}}(\RR^{3})$ and $\bT_{2} \in \Sym^{n_{2}}(\RR^{3})$, we can recast proposition~\ref{prop:symmetric-harmonic-decomposition} in tensorial form:
\begin{equation*}
  \Sym^{n}(\RR^{3})
  \simeq
  \HH^{n}(\RR^{3}) \oplus \HH^{n-2}(\RR^{3}) \oplus \dotsb \oplus \HH^{n-2r}(\RR^{3})
\end{equation*}
where $r=[n/2]$. More precisely, every totally symmetric tensor $\bT \in \Sym^{n}(\RR^{3})$ of order $n$ can be decomposed uniquely as
\begin{equation*}
  \bT = \bH_{0} + \Idd \odot \bH_{1} + \dotsb + \Idd^{\odot r-1}\odot \bH_{r-1}+ \Idd^{\odot r} \odot\bH_{r},
\end{equation*}
where $\bH_{k}$ is an harmonic tensor of degree $n-2k$ and $\Idd^{\odot k}$ means the symmetrized tensorial product of $k$ copies of $\Idd$. The harmonic part $(\bT)_{0}:=\bH_{0}$ is the orthogonal projection of $\bT$ onto $\HH^{n}(\RR^{3})$. Using~\eqref{eq:Lapl-k}, the tensorial form of~\eqref{eq:hr-poly} reads
\begin{equation*}
  \bH_{r} = \left\{
  \begin{array}{ll}
    \displaystyle{\frac{1}{2r+1}}\tr^r \bT, & \hbox{if $n$ is even;} \\
    \displaystyle{\frac{3}{2r+3}}\tr^r \bT, & \hbox{if $n$ is odd.}
  \end{array}
  \right.
\end{equation*}
which is a scalar for even $n$ or a vector for odd $n$. Finally, the tensorial form of~\eqref{eq:hk-poly} is given by:
\begin{equation*}
  \bH_{k}=\phi^{-1}(\rh_{k})=\mu(k) \frac{n!}{(n-2k)!} \tr^{k}\left[ \bT - \sum_{j=k+1}^{r} \Idd^{\odot j}\odot \bH_{j} \right]
\end{equation*}
Note that in this equation, $\tr^{k}\left[ \bT - \sum_{j=k+1}^{r} \Idd^{\odot j}\odot \bH_{j} \right]$ is the orthogonal projection of $\tr^{k} \bT$ on $\HH^{n-2k}(\RR^{3})$.

\subsection{Harmonic decomposition of the elasticity tensor}
\label{subsec:elasticity-tensor}

Harmonic decomposition of the space $\Ela$ of Elasticity tensors was first obtained by Backus~\cite{Bac1970} and is given by
\begin{equation*}
  \Ela
  \simeq
  \HH^{0}(\RR^{3})\oplus \HH^{0}(\RR^{3})\oplus \HH^{2}(\RR^{3})\oplus \HH^{2}(\RR^{3})\oplus \HH^{4}(\RR^{3}).
\end{equation*}

More precisely, an elasticity tensor $\tq{E}$ admits the following harmonic decomposition~\cite{Bac1970}:
\begin{equation}\label{eq:HarmE}
  \tq{E} = \alpha \, \Idd \otimes_{(4)} \Idd + \beta \, \Idd \otimes_{(2,2)}\! \Idd
  + \Idd \otimes_{(4)} \td a' + \Idd \otimes_{(2,2)} \! \td b' + \bH.
\end{equation}
In this formula, $\Idd$ is the Euclidean canonical bilinear 2-form represented by the components $(\delta_{ij})$ in any orthonormal basis and the Young-symmetrized tensor products $\otimes_{(4)}$ and $\otimes_{(2,2)}$, between symmetric second-order tensors $\ba,\bb$, are defined as follows:
\begin{equation*}
  (\td a \otimes_{(4)} \td b)_{ijkl} = \frac{1}{6} \big( a_{ij}b_{kl} + b_{ij} a_{kl} + a_{ik}b_{jl} + b_{ik}a_{jl} + a_{il}b_{jk} + b_{il}a_{jk} \big),
\end{equation*}
where $\otimes_{(4)}$ is the totally symmetric tensor product also denoted by $\odot$, and
\begin{equation*}
  (\td a \otimes_{(2,2)}\!\td b)_{ijkl} = \frac{1}{6} \big( 2a_{ij}b_{kl} + 2b_{ij}a_{kl} - a_{ik}b_{jl} - a_{il}b_{jk} - b_{ik}a_{jl} - b_{il}a_{jk} \big).
\end{equation*}
In the harmonic decomposition~\eqref{eq:HarmE}, $\alpha,\beta$ are scalars, $\td{a}',\td{b}'$ are second order harmonic tensors (also known as symmetric deviatoric tensors) and $\bH$ is a fourth-order harmonic tensor.

\begin{rem}
  In this decomposition, $\alpha,\beta$ and $\td{a}',\td{b}'$ are related to the \emph{dilatation tensor} $\td d:=\tr_{12} \tq E$ and the \emph{Voigt tensor} $\td v:=\tr_{13}\tq E$ by the formulas
  \begin{equation*}
    \alpha = \frac{1}{15}\left( \tr \bd + 2 \tr \bv\right)
    \qquad
    \beta = \frac{1}{6}\left( \tr \bd - \tr \bv\right)
  \end{equation*}
  and
  \begin{equation*}
    \td a'=\frac{2}{7} \left( \bd'+2 \bv ' \right)
    \qquad
    \td b'=2 \left( \bd'- \bv ' \right)
  \end{equation*}
  where $\td d':=\td d-\frac{1}{3}\tr( \td d)\,\Idd$ and $\td v':=\td v-\frac{1}{3}\tr(\td v)\,\Idd$ are the deviatoric parts of $\td d$ and $\td v$ respectively.
\end{rem}

\section{Harmonic factorization}
\label{sec:harmonic-factorization}

In this section, we introduce a product between harmonic tensors (which is associative and commutative). Then, we show that every harmonic tensor $\bH\in \HH^{n}(\RR^{3})$ can be factorized as an harmonic product of lower order tensors.

\subsection{Harmonic product}

The product of two harmonic polynomials is not harmonic in general. However, we define the following binary operation on the space of harmonic polynomials:
\begin{equation*}
  \rh_{1}\ast \rh_{2} := (\rh_{1}\rh_{2})_{0}.
\end{equation*}
This product is commutative:
\begin{equation*}
  \rh_{1}\ast \rh_{2}=\rh_{2}\ast \rh_{1},
\end{equation*}
and is moreover associative, which means that
\begin{equation*}
  (\rh_{1} \ast \rh_{2}) \ast \rh_{3} = \rh_{1} \ast (\rh_{2} \ast \rh_{3}).
\end{equation*}
The fact that this binary operation is associative is essential in practice. It means that the final result does not depend on the way we ``put the brackets" to compute it. We can write therefore
\begin{equation*}
  \rh_{1} \ast \rh_{2} \ast \rh_{3} := (\rh_{1} \ast \rh_{2}) \ast \rh_{3} = \rh_{1} \ast (\rh_{2} \ast \rh_{3}),
\end{equation*}
and omit the brackets. Endowed with this product, the vector space
\begin{equation*}
  \Hn{}(\RR^{3}) := \bigoplus_{n\geq 0} \Hn{n}(\RR^{3})
\end{equation*}
of harmonic polynomials becomes a \emph{commutative algebra}. The unity (for the multiplicative operation $\ast$) is represented by the constant harmonic polynomial $1$.

\begin{rem}\label{rem:H1odotH2}
  Using the equivariant isomorphism $\phi$ between $\HH^{n}(\RR^{3})$ and $\Hn{n}(\RR^{3})$, this algebraic structure is inherited on harmonic tensors. More precisely, given two harmonic tensors $\bH_{1}$ and $\bH_{2}$ of order $n_{1}$ and $n_{2}$ respectively, we define an harmonic tensor $\bH_{1} \ast \bH_{2}$ of order $n_{1}+n_{2}$ by the following formula:
  \begin{equation*}
    \bH_{1} \ast \bH_{2} := \phi^{-1}\left( \phi(\bH_{1}) \ast \phi(\bH_{2})\right) = \left(\bH_{1} \odot \bH_{2}\right)_{0}.
  \end{equation*}
\end{rem}

\begin{ex}[Harmonic product of two vectors]\label{ex:v-ast-w}
  For two vectors $\vv_{1},\vv_{2} \in \HH^{1}(\RR^{3})$, we get
  \begin{equation*}
    \vv_{1} \ast \vv_{2} = \frac{1}{2}\left(\vv_{1} \otimes \vv_{2} + \vv_{2} \otimes \vv_{1} \right) - \frac{1}{3} (\vv_{1} \cdot \vv_{2}) \Idd.
  \end{equation*}
\end{ex}

\begin{ex}[Harmonic product of two second order harmonic tensors]\label{ex:a-ast-b}
  For two second order harmonic tensors $\bh_{1}, \bh_{2} \in \HH^{2}(\RR^{3})$, we get
  \begin{equation*}
    \bh_{1} \ast \bh_{2} = \bh_{1} \odot \bh_{2} - \frac{2}{7}\, \Idd \odot (\bh_{1} \bh_{2} + \bh_{2} \bh_{1}) + \frac{2}{35} \tr(\bh_{1} \bh_{2})\, \Idd \odot \Idd.
  \end{equation*}
\end{ex}

\subsection{Harmonic factorization theorem}\label{subsec:Harm_Factor_Thm}

We now introduce the following $\SO(3)$-equivariant mapping, from the direct sum $\RR^{3} \oplus \dotsb \oplus \RR^{3}$ of $n$ copies of $\RR^{3}$ to $\Hn{n}(\RR^{3})$:
\begin{equation}\label{eq:Sylvester-map}
  \Phi : (\ww_{1}, \dotsc , \ww_{n}) \mapsto \rh(\xx) := (\xx \cdot \ww_{1}) \ast \dotsb \ast (\xx \cdot \ww_{n}).
\end{equation}
This mapping is obviously not \emph{one-to-one}. It is however \emph{surjective}. This result is not trivial and seems to have been demonstrated for the first time by Sylvester~\cite{Syl1909} (see also~\cite[Appendix 1]{Bac1970}). We provide our own proof of Sylvester's theorem in Appendix~\ref{sec:Annexe_Sylvester-theorem-proof}. It relies on \emph{binary forms}. Note however, that the use of binary forms in the proof limits the validity of this theorem to dimension $D = 3$ (see \cite{Bac1970} for a counter-example in dimension $D\geq 4$).

\begin{thm}[Sylvester's theorem]\label{thm:Sylvester}
  Let $\rh \in \Hn{n}(\RR^{3})$, then there exist $n$ vectors $\ww_{1}, \dotsc , \ww_{n}$ in $\RR^{3}$ such that
  \begin{equation*}
    \rh (\xx) = (\xx \cdot \ww_{1}) \ast \dotsb \ast (\xx \cdot \ww_{n}).
  \end{equation*}
\end{thm}

\begin{ex}\label{ex:multipoles-determination}
  Let $\rh \in \Hn{2}(\RR^{3})$ and $\ff\in \SRn{4}$ be the corresponding binary form (see example~\ref{ex:order2-harmonic-tensors}). The four roots of $\ff$ are
  \begin{equation*}
    \lambda_{1},-\frac{1}{\overline{\lambda_{1}}},\lambda_{2},-\frac{1}{\overline{\lambda_{2}}},
  \end{equation*}
  with the possible couple $(0,\infty)$ among them. Using the stereographic projection $\tau$ (see remark~\ref{ex:order1-harmonic-tensors}), we set
  \begin{equation*}
    \ww_{i} := \tau^{-1}(\lambda_{i}).
  \end{equation*}
  Note that
  \begin{equation*}
    \tau^{-1}\left(-\frac{1}{\overline{\lambda_{i}}}\right) = -\tau^{-1}(\lambda_{i}).
  \end{equation*}
  Then, we have
  \begin{equation*}
    \rh(\xx) = k (\xx \cdot \ww_{1}) \ast (\xx \cdot \ww_{2}),
  \end{equation*}
  for some constant $k \in \RR$.
\end{ex}

The lack of injectivity for the mapping $\Phi$ is detailed in the following proposition, which proof is provided in Appendix~\ref{sec:Annexe_Sylvester-theorem-proof}.

\begin{prop}\label{prop:fiber}
  We have
  \begin{equation*}
    \Phi(\tilde{\ww}_{1}, \dotsc , \tilde{\ww}_{n}) = \Phi(\ww_{1}, \dotsc , \ww_{n}).
  \end{equation*}
  if and only if
  \begin{equation*}
    (\tilde{\ww}_{1}, \dotsc , \tilde{\ww}_{n}) = (c_{1}\ww_{\sigma(1)}, \dotsc , c_{n}\ww_{\sigma(n)})
  \end{equation*}
  where $(\ww_{\sigma(1)}, \dotsc , \ww_{\sigma(n)})$ is a permutation of the vectors $(\ww_{1}, \dotsc , \ww_{n})$ for some $\sigma\in \perm_{n}$ and $c_{1}, \dotsc , c_{n}$ are real numbers such that $c_{1} \dotsb c_{n} = 1$.
\end{prop}

\begin{rem}
  It is tempting to restrict the mapping $\Phi$ defined in~\eqref{eq:Sylvester-map} to unit vectors, in order to try to reduce the indeterminacy in the choice of an $n$-uple of vectors $(\ww_{1}, \dotsc , \ww_{n})$. However, this does not really solve the problem of indeterminacy. Of course, the indeterminacy due to the scaling by the $c_{i}$ is replaced by signs $\epsilon_{i}$, but the indeterminacy due to the permutation persists anyway.
\end{rem}

Sylvester's theorem is equivalent to the following (in appearance, more general) theorem which leads to other important factorization results for harmonic polynomials (and thus harmonic tensors).

\begin{thm}[Harmonic factorization theorem]\label{thm:harmonic-factorization-theorem}
  Let $\rh \in \Hn{kn}(\RR^{3})$. Then there exist harmonic polynomials $\rh_{1}, \dotsc, \rh_{k} \in \Hn{n}(\RR^{3})$ such that
  \begin{equation*}
    \rh = \rh_{1} \ast \dotsb \ast \rh_{k}.
  \end{equation*}
\end{thm}

\begin{rem}\label{rem:Tensorial-Hkn}
  The tensorial counterpart of theorem~\ref{thm:harmonic-factorization-theorem} is that for every harmonic tensor
  $\bH \in \HH^{kn}(\RR^{3})$ there exist harmonic tensors $\bH_{1},\dots, \bH_{k}\in \HH^n(\RR^{3})$ such that
  \begin{equation*}
    \bH= \bH_{1} \ast \dotsb \ast \bH_{k}.
  \end{equation*}
\end{rem}

A corollary of theorem~\ref{thm:harmonic-factorization-theorem} is the following result, which applies to every even order harmonic tensor.

\begin{cor}\label{cor:diffcarre_tensor}
  Let $\bH\in \HH^{2n}(\RR^{3})$. Then there exist $\bH_{1}, \bH_{2}\in \HH^{n}(\RR^{3})$ such that
  \begin{equation}\label{eq:square-decomposition}
    \bH = \bH_{1}\ast \bH_{1} -\bH_{2}\ast \bH_{2}.
  \end{equation}
\end{cor}

\begin{proof}
  Let $\rh \in \Hn{2n}(\RR^{3})$ corresponding to the harmonic tensor $\mathbf{H}$. Using theorem~\ref{thm:harmonic-factorization-theorem}, we can find $\rp_{1}, \rp_{2} \in \Hn{n}(\RR^{3})$ such that $\rh = \left(\rp_{1} \rp_{2}\right)_{0}$. But
  \begin{equation*}
    \rp_{1}\rp_{2} = \left(\frac{\rp_{1}+\rp_{2}}{2}\right)^{2}-\left(\frac{\rp_{1}-\rp_{2}}{2}\right)^{2}.
  \end{equation*}
  Hence, if we set $\rh_{1} = (\rp_{1}+\rp_{2})/2$ and $\rh_{2} = (\rp_{1}-\rp_{2})/2$, we get
  \begin{equation*}
    \rh = \rh_{1} \ast \rh_{1} - \rh_{2} \ast \rh_{2},
  \end{equation*}
  which tensorial form is precisely~\eqref{eq:square-decomposition}.
\end{proof}

\begin{rem}
  This result has already been obtained in~\cite{DD2016} for monoclinic fourth-order harmonic tensors, using direct elimination techniques.
\end{rem}

\begin{rem}
  In 2D, the factorization of the harmonic part $\bH^{2D}\in \HH^{4}(\RR^{2})$ of the elasticity tensor has already been computed in \cite{DD2015}, using previous results obtained in~\cite{Ver1979,Van2005}. In that case, we get
  \begin{equation}\label{eq:H2Dh0}
    \bH^{2D}=\bh\otimes \bh -\frac{1}{2}({\bh}\cdot{\bh})\; \tq J^{2D}
  \end{equation}
  where $\bh$ is a second order harmonic tensor ({\it i.e.} a symmetric deviatoric), and where $ \tq J^{2D}=\tq I -\frac{1}{2} \, \Idd \otimes \Idd$ and ${\bh}\cdot{\bh}=\| \bh\|^{2}$. As in 3D, we can define an harmonic product between 2D harmonic tensors $\bh_k$
  \begin{equation*}
    \bh_{1}\ast\bh_{2}:=\bh_{1}\odot \bh_{2}-\frac{1}{4}\tr (\bh_{1}\bh_{2})\Idd \odot \Idd
  \end{equation*}
  and equation~\eqref{eq:H2Dh0} can be rewritten as
  \begin{equation}\label{eq:square2D}
    \bH^{2D} = \bh\ast\bh.
  \end{equation}
  Hence, any 2D fourth-order harmonic tensor can be expressed as the harmonic square $\bh\ast\bh$, \textit{i.e.}  the projection of diadic symmetric square $\bh \otimes \bh$ onto the vector space of fourth-order harmonic tensors $\HH^{4}(\RR^{2})$. In 3D, this result is generally false. Indeed, if a fourth-order harmonic tensor $\bH$ is an harmonic square, \textit{i.e.} $\bH = \bh \ast \bh$,
  then $\bH$ is either orthotropic or transversely isotropic (see section~\ref{sec:reconstruction}).
\end{rem}

Although these various decompositions are interesting by themselves, there is however something frustrating. These decompositions are not unique in general and the solution is not enough constructive (they require root's extraction of possible high degree polynomials). In the next sections, we provide explicit geometric reconstructions.

\section{Equivariant reconstruction of a fourth-order harmonic tensor}
\label{sec:reconstruction}

In this section, we consider the general problem of reconstructing a fourth-order harmonic tensor $\bH$ using lower order \emph{covariants} (i.e. lower order tensors which depend on $\bH$ in an equivariant manner). This reconstruction is useful when applied to the fourth-order harmonic part $\bH$ of the elasticity tensor (\ref{eq:HarmE}), since its other components $\alpha, \beta, \ba', \bb'$ are already decomposed by means of second order covariants tensors and invariants.

From a mechanical point of view the definition of covariants of a given effective --for example damaged-- elasticity tensor is strongly related to the tensorial nature of the thermodynamics variables allowing to properly model a specific state coupling with elasticity~\cite{LO1980,Ona1984,LC1985,BHL1995,LD2005}. In the case of effective elasticity of cracked/damaged media for instance, either second order \cite{Kac1972,CS1982,Lad1983} or fourth order \cite{Cha1979,LO1980} damage variables are considered, depending on the authors choice \cite{LC1985,Kra1996,LD2005}.
Considering the reconstruction by means of lower order covariants of the harmonic part of an effective elasticity tensor in this case would answer the question: do we have to use a fourth order damage variable -- a fourth order covariant in next derivations -- in order to reconstruct a given harmonic -- effective elasticity -- tensor? In 2D the answer is negative, as shown in \cite{DD2015,DD2016}: only one second order tensor $\bh$ is needed (see Eq.~\eqref{eq:square2D}).

\subsection{The reconstruction problem}

To illustrate the problem, let us first come back to theorem~\ref{thm:harmonic-factorization-theorem}. Applied to a fourth-order harmonic tensor $\bH$, this theorem states that $\bH$ can be factorized as the harmonic product of two second order harmonic tensors
\begin{equation*}
  \bH = \bh_{1} \ast \bh_{2}.
\end{equation*}
In other words, the equivariant mapping
\begin{equation*}
  \mathcal{F} : \HH^{2}(\RR^{3}) \times \HH^{2}(\RR^{3}) \to \HH^{4}(\RR^{3}), \qquad (\bh_{1}, \bh_{2}) \mapsto  \bh_{1} \ast \bh_{2}
\end{equation*}
is surjective. What would be more interesting is to have an explicit mapping (i.e. an explicit construction of the solution)
\begin{equation*}
  \mathcal{S} : \HH^{4}(\RR^{3}) \to \HH^{2}(\RR^{3}) \times \HH^{2}(\RR^{3}), \qquad \bH \mapsto  (\bh_{1}(\bH), \bh_{2}(\bH))
\end{equation*}
such that:
\begin{equation*}
  (\mathcal{F} \circ \mathcal{S})(\bH) = \bh_{1}(\bH) \ast \bh_{2}(\bH) = \bH, \qquad \forall \bH \in \HH^{4}(\RR^{3}),
\end{equation*}
and with nice geometric properties. A mapping like $\mathcal{S}$ is called a \emph{section} of $\mathcal{F}$. A nice geometric property would be for instance for $\mathcal{S}$ to be equivariant, in which case each second order tensors $\bh_{i}(\bH)$ is a covariant tensor. The bad news is that there does not exist any \emph{equivariant section} $\mathcal{S}$ for $\mathcal{F}$ defined on the whole space $\HH^{4}(\RR^{3})$: for instance, for a tensor $\bH$ with cubic symmetry, all second order harmonic covariants vanish, which forbids the existence of such a general equivariant section. We will provide now the following definitions.

\begin{defn}[Decomposition]
  Let $\bT$ be some $n$th order tensor. A \emph{decomposition} of $\bT$ into other \emph{tensors} (possibly scalars) $\bT_k$ ($k=1,\dotsc,N$) is a mapping
  \begin{equation*}
    \bT = \mathcal{F} (\bT_{1},\dotsc,\bT_{N}).
  \end{equation*}
  The decomposition is \emph{equivariant} if $\mathcal{F}$ is an equivariant mapping.
\end{defn}

\begin{defn}[Reconstruction]
  Given a decomposition
  \begin{equation*}
    \bT = \mathcal{F} (\bT_{1},\dotsc,\bT_{N})
  \end{equation*}
  of a tensor $\bT$ into other \emph{tensors} (possibly scalars), a \emph{reconstruction} of $\bT$ is a section $\mathcal{S}$ of $\mathcal{F}$. The reconstruction is \emph{equivariant} if both $\mathcal{F}$ and $\mathcal{S}$ are equivariant mappings.
\end{defn}

\begin{rem}
  In other words, a reconstruction is given by a mapping
  \begin{equation*}
    \mathcal{S}: \bT \mapsto (\kappa_{1}(\bT),\dotsc,\kappa_{N}(\bT))
  \end{equation*}
  such that
  \begin{equation*}
    \bT = \mathcal{F} (\kappa_{1}(\bT),\dotsc,\kappa_{N}(\bT)).
  \end{equation*}
  If the reconstruction is equivariant the $\bT_k=\kappa_{k}(\bT)$ are \emph{covariant tensors} (or invariants if there are scalars), which means that
  \begin{equation*}
    g\star \bT_k=\kappa_k(g\star \bT),\quad \forall g\in \SO(3).
  \end{equation*}
\end{rem}

\subsection{Obstruction to an equivariant reconstruction}\label{subsec:Obstruction_Equiv_reconstruction}

The existence of an equivariant reconstruction may not be possible for certain symmetry classes. More precisely, we have the following result.

\begin{lem}\label{lem:obstruction-equivariant-reconstruction}
  Consider an equivariant decomposition $\bT = \mathcal{F}(\bT_{1},\dotsc,\bT_{N})$ and suppose that there exists an equivariant reconstruction $\mathcal{S}$, so that $\bT_{k}$ are 
    {covariant} tensors. Then
  \begin{equation}\label{eq:Sym_Group}
    G_{\bT} = \bigcap_{k} G_{\bT_{k}},
  \end{equation}
  where
  \begin{equation*}
    G_{\bT} := \set{g\in G;\; g\star\bT = \bT}
  \end{equation*}
  is the \emph{symmetry group} of the tensor $\bT$.
\end{lem}

\begin{proof}
  Consider an equivariant decomposition $\bT = \mathcal{F}(\bT_{1},\dotsc,\bT_{N})$. Then, we have
  \begin{equation*}
    \bigcap_{k} G_{\bT_{k}}\subset G_{\bT}.
  \end{equation*}
  If moreover $\bT_{k}=\kappa_{k}(\bT)$ are covariant tensors, we get thus
  \begin{equation*}
    g\star\bT_{k} = \kappa_{k}(g \star \bT) = \kappa_{k}(\bT) = \bT_{k},
  \end{equation*}
  for all $g\in G_{\bT}$, and hence $G_{\bT}\subset G_{\bT_{k}}$, which achieves the proof.
\end{proof}

According to lemma~\ref{lem:obstruction-equivariant-reconstruction}, the existence of an equivariant reconstruction associated to an equivariant decomposition
\begin{equation*}
  \bT = \mathcal{F} (\kappa_{1}(\bT),\dotsc,\kappa_{N}(\bT))
\end{equation*}
requires some conditions on the symmetry groups of the involved covariants $\bT_{k}=\kappa_{k}(\bT)$. For a decomposition involving
  {symmetric} second-order covariant tensors, we have the following result (for details on closed subgroups of $\SO(3)$, see Appendix~\ref{sec:symmetry-classes}).

\begin{cor}\label{cor:second-order-obstruction}
  If $\bT = \mathcal{F}(\bT_{1},\dotsc,\bT_{N})$ is an equivariant decomposition of $\bT$ into \emph{second-order symmetric covariant tensors} $\bT_{k}$, then
  \begin{equation*}
    [G_{\bT}]\in \set{[\triv],[\ZZ_{2}],[\DD_{2}],[\OO(2)],[\SO(3)]},
  \end{equation*}
  where $[H]$ means the conjugacy class of a subgroup $H \subset \SO(3)$.
\end{cor}

\begin{proof}
  The proof is easily achieved using the \emph{clips operator} introduced in~\cite{Oli2014}. This binary operator, noted $\circledcirc$ is defined on conjugacy classes (or finite families of conjugacy classes) of closed subgroups of $\SO(3)$:
  \begin{equation*}
    [H_{1}]\circledcirc [H_{2}]=\set{ [H_{1}\cap (gH_{2}g^{-1})],\quad g\in \SO(3)}.
  \end{equation*}
  It is an associative and commutative operation and
  \begin{equation*}
    [H] \circledcirc [\triv] = \set{[\triv]}, \qquad [H] \circledcirc [\SO(3)] = \set{[H]},
  \end{equation*}
  for all closed subgroup $H$. Moreover, we have in particular (see~\cite{Oli2014}):

  \begin{center}
    \begin{tabular}{|c|c|c|c|}
      \hline
      $\circledcirc$ & $[\ZZ_{2}]$               & $[\DD_{2}]$                         & $[\OO(2)]$                           \\
      \hline
      $[\ZZ_{2}]$    & $\set{[\triv],[\ZZ_{2}]}$ & $\set{[\triv],[\ZZ_{2}]}$           & $\set{[\triv],[\ZZ_{2}]}$            \\
      \hline
      $[\DD_{2}]$    &                           & $\set{[\triv],[\ZZ_{2}],[\DD_{2}]}$ & $\set{[\triv],[\ZZ_{2}],[\DD_{2}]}$  \\
      \hline
      $[\OO(2)]$     &                           &                                     & $\set{[\ZZ_{2}],[\DD_{2}],[\OO(2)]}$ \\
      \hline
    \end{tabular}
  \end{center}

  Thus, if each $\bT_{k}$ is a second-order symmetric covariant (or an invariant), then $[G_{\bT_{k}}]$ is either $[\DD_{2}]$, $[\OO(2)]$ or $[\SO(3)]$ (see ~\Cref{sec:symmetry-classes}) and we deduce that
  \begin{equation*}
    [G_{\bT}] \in [H_{1}] \circledcirc [H_{2}] \circledcirc \dotsb \circledcirc [H_{N}],
  \end{equation*}
  where each $[H_{i}]$ belongs to $\set{[\DD_{2}],[\OO(2)],[\SO(3)]}$, which achieves the proof.
\end{proof}

\begin{rem}
  Corollary ~\ref{cor:second-order-obstruction} implies, in particular, that there is no equivariant reconstruction by means of second order symmetric covariant tensors of a fourth order harmonic tensor $\bH$ (and thus of the elasticity tensor $\tq E$) belonging to the cubic, the trigonal or the tetragonal symmetry class. For a cubic tensor, the situation is even worse since all its second order symmetric covariants are spheric, whereas its first order covariants vanish.
\end{rem}

\begin{rem}
  The geometric framework introduced here is connected to the so-called \emph{isotropic extension} of anisotropic constitutive functions via structural tensors developed by Boehler~\cite{Boe1979} and Liu~\cite{Liu1982} independently (see also~\cite{Man2016}). In the case of a linear constitutive function, an isotropic extension is just an equivariant decomposition \emph{limited to a given symmetry class} and for which the arguments $(\bT_{1},\dotsc,\bT_{N})$ (the structural tensors) of the decomposition satisfy~\eqref{eq:Sym_Group}. In that case, our approach is however more constructive since we are able to give an \emph{explicit equivariant formula} in which the ``structural tensors'' are covariant tensors of the constitutive tensor. Moreover, condition~\eqref{eq:Sym_Group} is a corollary of the theory and does not need to be imposed \emph{a priori} as an hypothesis.
\end{rem}

\section{Reconstructions by means of second order covariants}
\label{sec:second-order-covariants}

In this section, we provide explicit equivariant reconstructions for transverse-isotropic and orthotropic fourth order harmonic tensors. These reconstructions use the following second-order (polynomial) covariants introduced in~\cite{BKO1994}:
\begin{equation}\label{eq:Boehler-covariants}
  \begin{array} {lll}
    \td{d}_{2} := \tr_{13} \tq{H}^{2},                    & \td{d}_{3} := \tr_{13} \tq{H}^{3},                    & \td{d}_{4} := \td{d}_{2}^{2},
    \\
    \td{d}_{5} := \td{d}_{2} (\tq{H} \td{d}_{2}),         & \td{d}_{6} := \td{d}_{2}^{3},                         & \td{d}_{7} := \td{d}_{2}^{2} (\tq{H} \td{d}_{2})
    \\
    \td{d}_{8} := \td{d}_{2}^{2} (\tq{H}^{2} \td{d}_{2}), & \td{d}_{9} := \td{d}_{2}^{2} (\tq{H} \td{d}_{2}^{2}), & \td{d}_{10} := \td{d}_{2}^{2} (\tq{H}^{2} \td{d}_{2}^{2}).
  \end{array}
\end{equation}
where the following notations have been used.
\begin{itemize}
  \item For two fourth-order tensors $\tq{A}$ and $\tq{B}$, $\tq{A}\tq{B}$ is the fourth-order tensor
        \begin{equation*}
          (\tq{A}\tq{B})_{ijkl}:=A_{ijpq}B_{pqkl}.
        \end{equation*}
  \item For a fourth-order tensor $\tq{A}$ and a second-order tensor $\tq{b}$, $\tq{A}\td{b}$ is the second-order tensor
        \begin{equation*}
          (\tq{A}\td{b})_{ij}:=A_{ijkl}b_{kl}.
        \end{equation*}
\end{itemize}
We also make use of the following invariants:
\begin{equation}\label{eq:Boehler-invariants}
  J_{k} := \tr \td{d}_{k} , \qquad k = 2, \dotsc ,10,
\end{equation}
which constitute an \emph{integrity basis} for $\HH^{4}(\RR^{3})$ (see~\cite{BKO1994}).

\begin{rem}
  The second--order covariants $\td{d}_{k}$ ($k=2,3,4,6$) are always \emph{symmetric} covariants, while in general $\td{d}_{k}$ ($k=5,7,8,9,10$) are neither symmetric nor antisymmetric.
\end{rem}

Finally, using the Kelvin representation, a fourth order harmonic tensor $\bH=(H_{ijkl})$ is represented by the symmetric matrix
\begin{equation*}
  \underline{\bH} =
  \begin{pmatrix}
    H_{1111}         & H_{1122}         & H_{1133}         & \sqrt{2}H_{1123} & \sqrt{2}H_{1113} & \sqrt{2}H_{1112} \\
    H_{1122}         & H_{2222}         & H_{2233}         & \sqrt{2}H_{2223} & \sqrt{2}H_{1223} & \sqrt{2}H_{1222} \\
    H_{1133}         & H_{2233}         & H_{3333}         & \sqrt{2}H_{2333} & \sqrt{2}H_{1333} & \sqrt{2}H_{1233} \\
    \sqrt{2}H_{1123} & \sqrt{2}H_{2223} & \sqrt{2}H_{2333} & 2H_{2233}        & 2H_{1233}        & 2H_{1223}        \\
    \sqrt{2}H_{1113} & \sqrt{2}H_{1223} & \sqrt{2}H_{1333} & 2H_{1233}        & 2H_{1133}        & 2H_{1123}        \\
    \sqrt{2}H_{1112} & \sqrt{2}E_{1222} & \sqrt{2}H_{1233} & 2H_{1223}        & 2H_{1123}        & 2H_{1122}
  \end{pmatrix}
\end{equation*}
with only 9 independent components due to the traceless property, namely:
\begin{align*}
  H_{1111} & = - H_{1122} - H_{1133}, & H_{2222} & = -H_{1122}-H_{2233},
  \\
  H_{3333} & = - H_{1133} - H_{2233}, & H_{2333} & = - H_{1123} - H_{2223},
  \\
  H_{1113} & = - H_{1223} - H_{1333}, & H_{1222} & = - H_{1112} - H_{1233}.
\end{align*}

\subsection{The transversely isotropic class}
\label{subsec:transversely-isotropic}

An harmonic tensor $\bH\in \HH^{4}(\RR^{3})$ is transverse isotropic if and only if there exists $g\in \SO(3)$ such that $\bH = g\star \bH_{0}$ where $\bH_{0}$ has the following normal matrix form
\begin{equation*}
  \underline{\bH_{0}} = \left(
  \begin{array}{cccccc}
      3 \delta  & \delta    & -4 \delta & 0         & 0         & 0        \\
      \delta    & 3 \delta  & -4 \delta & 0         & 0         & 0        \\
      -4 \delta & -4 \delta & 8 \delta  & 0         & 0         & 0        \\
      0         & 0         & 0         & -8 \delta & 0         & 0        \\
      0         & 0         & 0         & 0         & -8 \delta & 0        \\
      0         & 0         & 0         & 0         & 0         & 2 \delta
    \end{array}
  \right),
\end{equation*}
where $\delta \ne 0$. It was moreover established in~\cite[Section 5.2]{AKP2014} that
\begin{equation*}
  J_{2}\ne 0, \qquad \delta=\frac{7J_{3}}{18J_{2}}.
\end{equation*}
Thus, by a direct computation of the second order covariant $\bd_{2}$ (see Eq.~\eqref{eq:Boehler-covariants}) and its deviatoric part $\bd_{2}^{\prime}$, we get the following result:

\begin{thm}\label{thm:O2}
  For any transversely isotropic fourth-order harmonic tensor $\bH$, we have
  \begin{equation*}
    \bH = \frac{63}{25}\frac{1}{J_{3}(\bH)} \, \bd_{2}^{\prime}(\bH) \ast \bd_{2}^{\prime}(\bH),
  \end{equation*}
\end{thm}

\begin{rem}\label{rem:O2-harmonic-square}
  If $J_{3}(\bH) > 0$, then $\bH$ is a perfect harmonic square. The converse is also true and can be easily established, using binary forms, where the harmonic product corresponds to the ordinary product of polynomials.
\end{rem}

\subsection{The orthotropic class}
\label{subsec:orthotropic}

An harmonic tensor $\bH\in \HH^{4}(\RR^{3})$ is orthotropic if and only if there exists $g\in \SO(3)$ such that $\bH = g\star \bH_{0}$ where $\bH_{0}$ has the following normal matrix form
\begin{equation}\label{eq:normal-form-D2}
  \underline{\bH_{0}}= \left(
  \begin{array}{cccccc}
      \lambda_{2} + \lambda_{3} & -\lambda_{3}              & -\lambda_{2}              & 0               & 0               & 0               \\
      -\lambda_{3}              & \lambda_{3} + \lambda_{1} & -\lambda_{1}              & 0               & 0               & 0               \\
      -\lambda_{2}              & -\lambda_{1}              & \lambda_{2} + \lambda_{1} & 0               & 0               & 0               \\
      0                         & 0                         & 0                         & -2\,\lambda_{1} & 0               & 0               \\
      0                         & 0                         & 0                         & 0               & -2\,\lambda_{2} & 0               \\
      0                         & 0                         & 0                         & 0               & 0               & -2\,\lambda_{3}
    \end{array}
  \right)
\end{equation}
where $\lambda_{1},\lambda_{2},\lambda_{3}$ are three \emph{distinct} real numbers. Note that this normal form is however not unique: any permutation of the $\lambda_{k}$ provides an alternative normal form, but this is the only ambiguity. It is therefore useful to introduce the three elementary symmetric functions
\begin{equation}\label{eq:sym-func}
  \begin{split}
    \sigma_{1} & :=\lambda_{1}+\lambda_{2}+\lambda_{3},\\
    \sigma_{2} & :=\lambda_{1}\lambda_{2}+\lambda_{1}\lambda_{3}+\lambda_{2}\lambda_{3}, \\
    \sigma_{3} & :=\lambda_{1}\lambda_{2}\lambda_{3}
  \end{split}
\end{equation}
which are independent of a particular normal form, as being invariant under any permutation of the $\lambda_i$. Moreover, it was shown in~\cite[Section 5.5]{AKP2014} that the \emph{discriminant}
\begin{equation*}
  \Delta_{3}:=(\lambda_{2} - \lambda_{1})^{2}(\lambda_{2} - \lambda_{3})^{2}(\lambda_{3} - \lambda_{1})^{2},
\end{equation*}
which is \emph{strictly positive}, is a polynomial invariant, and that the $\sigma_{k}$ are themselves \emph{rational invariants}. More precisely:
\begin{equation*}
  \Delta_{3} = \frac{K_{6}}{432}, \quad \text{where} \quad K_{6} := 6J_{6} - 9J_{2}J_{4} - 20J_{3}^{2} + 3J_{2}^{3},
\end{equation*}
and
\begin{equation*}
  \begin{split}
    \sigma_{1} &= \displaystyle \frac{9 (3J_{7} - 3J_{2}J_{5} + 3J_{3}J_{4} - J_{2}^{2}J_{3})}{2K_{6}},
    \\
    \sigma_{2} &= \displaystyle\frac{4}{7}\sigma_{1}^{2}-\frac{1}{14}J_{2},
    \\
    \sigma_{3}& = \displaystyle -\frac{1}{24}J_{3}+\frac{1}{7}\sigma_{1}^{3}-\frac{1}{56}\sigma_{1}J_{2}.
  \end{split}
\end{equation*}

Consider now the matrix
\begin{equation*}
  \blambda_{0} :=
  \begin{pmatrix}
    \lambda_{1} & 0           & 0           \\
    0           & \lambda_{2} & 0           \\
    0           & 0           & \lambda_{3}
  \end{pmatrix}.
\end{equation*}

Since both $\bH_{0}$ and $\blambda_{0}$ have the same symmetry group, namely $\DD_{2}$ (as defined in \Cref{sec:symmetry-classes}), the following definition
\begin{equation*}
  \blambda(\bH) := g \star \blambda_{0}, \quad \text{if} \quad \bH = g \star \bH_{0}
\end{equation*}
does not depend on the rotation $g$ and thus defines a covariant mapping from the orthotropic class in $\HH^{4}(\RR^{3})$ to the space of symmetric second order tensors. We have moreover the following result, which can be checked by a direct evaluation on $\bH_{0}$ and $\blambda_{0}'$.

\begin{lem}\label{lem:rational-expression-lambda}
  Let $\bH$ be an orthotropic fourth-order harmonic tensor. Then, the deviatoric part $\blambda'(\bH)$ of the second order covariant $\blambda(\bH)$ can be written as
  \begin{equation}\label{eq:widetildelambda}
    \blambda'(\bH) =  \frac{1}{8\Delta_{3}} \left(\alpha_{2}\td d_{2}'(\bH) +\alpha_{3}\td d_{3}'(\bH) -54 \sigma_{3}\,\td d_{4}'(\bH) +11 \sigma_{2}\,\td d_{5}'(\bH)\right)
  \end{equation}
  where
  \begin{align*}
    \alpha_{2} & :=2 (112 \sigma_{1}^{2} \sigma_{3}+21 \sigma_{1} \sigma_{2}^{2}-270 \sigma_{2} \sigma_{3}), \\
    \alpha_{3} & := 8(14 \sigma_{1} \sigma_{3}-11 \sigma_{1}^{2} \sigma_{2}+15 \sigma_{2}^{2}).
  \end{align*}
\end{lem}

Note that
\begin{equation*}
  2(\sigma_{1}^{2} - 3\sigma_{2})=(\lambda_{1}-\lambda_{2})^{2}+(\lambda_{1}-\lambda_{3})^{2}+(\lambda_{2}-\lambda_{3})^{2} > 0,
\end{equation*}
and we can thus introduce the positive invariant
\begin{equation*}
  \sigma_{eq} := \sqrt{\sigma_{1}^{2} - 3\sigma_{2}}.
\end{equation*}
and the \emph{Lode invariant} defined by
\begin{equation*}
  \mathcal{L} : = \frac{1}{\sigma_{eq}^{3}}\left(\sigma_{1}^{3}-\frac{9}{2} \sigma_{1}\sigma_{2} + \frac{27}{2}\sigma_{3}\right).
\end{equation*}
Since moreover
\begin{align*}
  \Delta_{3} & = \sigma_{1}^{2} \sigma_{2}^{2}+18 \sigma_{1} \sigma_{2} \sigma_{3}-4 \sigma_{1}^{3} \sigma_{3}-4 \sigma_{2}^{3}-27 \sigma_{3}^{2} \\
             & = \frac{4}{27} \sigma_{eq}^{6}(1 - \mathcal{L}^{2}),
\end{align*}
and $\Delta_{3} > 0$ for present orthotropic class, we deduce that $-1<\mathcal{L}<1$.

With these notations, we obtain the following result which can be checked by a direct computation on the normal form.

\begin{thm}\label{thm:D2}
  For any orthotropic fourth-order harmonic tensor $\bH$, we have
  \begin{equation} \label{eq:ThmReconstOrth}
    \bH = h_{1}\,\blambda'(\bH)\ast \blambda'(\bH) + 2 h_{2}\, \blambda'(\bH) \ast (\blambda'(\bH)^{2})' +h_{3}\,(\blambda'(\bH)^{2})' \ast (\blambda'(\bH)^{2})' ,
  \end{equation}
  where $\blambda'(\bH)$ is defined by~\eqref{eq:widetildelambda} and the three invariants $h_k$ are given by
  \begin{equation*}
    h_{1} :=  \frac{5 \sigma_{1} + 7 \mathcal{L} \sigma_{eq}}{2(1-\mathcal{L}^{2}) \sigma_{eq}^{2}},\quad
    h_{2}:=-\frac{3(5 \mathcal{L} \sigma_{1} + 7 \sigma_{eq} )}{2(1-\mathcal{L}^{2}) \sigma_{eq}^{3}},\quad
    h_{3}:=\frac{9(5 \sigma_{1} + 7 \mathcal{L} \sigma_{eq} )}{2(1-\mathcal{L}^{2}) \sigma_{eq}^{4}}.
  \end{equation*}
\end{thm}

\begin{rem}
  The three invariants $h_k$ are in fact rational invariants of $\bH$. Using invariants $\sigma_{k}$ (Eq. (\ref{eq:sym-func})), we have indeed
  \begin{equation*}
    \left\{
    \begin{aligned}
      h_{1} & =
      \frac{\left(\sigma_{1}^{2}-3 \sigma_{2}\right) \left(8 \sigma_{1}^{3}-31 \sigma_{1}
        \sigma_{2}+63 \sigma_{3}\right)}{9 \Delta_{3}},
      \\
      h_{2} & =-\frac{16 \sigma_{1}^4-86 \sigma_{1}^{2} \sigma_{2}+90 \sigma_{1} \sigma_{3}+84
        \sigma_{2}^{2}}{6\Delta_{3}},
      \\
      h_{3} & =\frac{8 \sigma_{1}^{3}-31 \sigma_{1} \sigma_{2}+63 \sigma_{3}}{\Delta_{3}}.
    \end{aligned}
    \right.
  \end{equation*}
\end{rem}

As in the transversely isotropic case, the following theorem characterized orthotropic harmonic tensors which are perfect harmonic squares.

\begin{thm}\label{thm:Orth_Square}
  An orthotropic fourth-order harmonic tensor $\bH$ is a perfect harmonic square $\bH=\bh \ast \bh$ if and only if
  \begin{equation}\label{eq:Cond_Orth_Square}
    \sigma_{1}> 0 \quad \text{and} \quad 49 \sigma_{2}-8 \sigma_{1}^{2}=0.
  \end{equation}
  In that case the second order harmonic tensor $\bh$ is given by
  \begin{equation}\label{eq:expr_racine_carre_Ortho}
    \bh := \pm \sqrt{\frac{49}{10(1-\mathcal{L})\sigma_{1}}}\;\left(\blambda'(\bH) - \frac{21}{5\sigma_{1}}
    (\blambda'(\bH)^{2})'\right).
  \end{equation}
\end{thm}

\begin{proof}
  Let $\bH$ be an orthotropic fourth order harmonic tensor. Suppose first that conditions~\eqref{eq:Cond_Orth_Square} are satisfied. Then, the coefficients $h_{1}$, $h_{2}$, $h_{3}$ in~\eqref{eq:ThmReconstOrth} satisfy
  \begin{equation*}
    h_{2}^{2} - h_{1}h_{3} = \frac{8 \sigma_1^2 -49 \sigma_2}{\Delta_3^2} = 0
  \end{equation*}
  and we have thus a factorization $\bH=\bh \ast \bh$, where $\bh$ is defined by~\eqref{eq:expr_racine_carre_Ortho}.

  Conversely, suppose that $\bH$ can be factorized as $\bH=\bh \ast \bh$. Without loss of generality, we can assume that $\bh$ is diagonal, with diagonal elements $a_{1}$, $a_{2}$ and $-(a_{1}+a_{2})$. Thus $\bH$ is fixed by $\DD_{2}$ and is represented by matrix normal form~\eqref{eq:normal-form-D2}. Consider now the associated binary form $\bg$ of $\bh$ (resp. $\ff$ of $\bH$) -- see \Cref{sec:Annexe_binary-forms}. We get then
  \begin{equation*}
    \bg=\alpha_{0}u^4+\alpha_{2}u^{2}v^{2}+\alpha_{0}v^{2},\quad \alpha_{0}:=\frac{1}{4}(a_{0}-a_{1}),\quad \alpha_{2}=-\frac{3}{2}(a_{0}+a_{1}),
  \end{equation*}
  and
  \begin{multline*}
    \ff =  \frac{1}{16}(\lambda_{1}+\lambda_{2}+8 \lambda_{3})u^8+\frac{7}{4} u^6 v^{2} (\lambda_{1}-\lambda_{2})+\frac{35}{8} u^{4} v^{4} (\lambda_{1}+\lambda_{2}) \\
    +\frac{7}{4} (\lambda_{1}-\lambda_{2}) u^{2} v^6 +\frac{1}{16}(\lambda_{1}+\lambda_{2}+8 \lambda_{3})v^8.
  \end{multline*}
  Now, $\ff=\bg^{2}$ leads to
  \begin{equation*}
    \begin{cases}
      \alpha_{2}^{2}+ 2 \alpha_{0}^{2} = \frac{35}{8}(\lambda_{1}+\lambda_{2})\geq 0,
      \\
      \alpha_{0} \alpha_{2} = \frac{7}{8} (\lambda_{1}-\lambda_{2}),
      \\
      \alpha_{0}^{2} = \frac{1}{16}(\lambda_{1}+\lambda_{2}+8 \lambda_{3}) \geq 0,
    \end{cases}.
  \end{equation*}
  and we get
  \begin{equation*}
    \sigma_{1} = \frac{{{\alpha_{2}}^{2}}+12{{\alpha_{0}}^{2}}}{5}, \qquad \sigma_{2} = \frac{8{{\alpha_{2}}^{4}}+192{{\alpha_{0}}^{2}}\,{{\alpha_{2}}^{2}}+1152{{\alpha_{0}}^{4}}}{1225}
  \end{equation*}
  Therefore, we have
  \begin{equation*}
    \sigma_{1}\geq 0, \qquad 8\sigma_{1}^{2}-49\sigma_{2}=0.
  \end{equation*}
  If $\sigma_{1}=0$, we get also $\sigma_{2}=0$ and thus $\Delta_{3} = -243\sigma_{3}^{2} \leq 0$, which leads to a contradiction (since $\bH$ is supposed to be orthotropic). This concludes the proof.
\end{proof}

\section{Reconstructions by means of second order covariants with a cubic covariant remainder}
\label{sec:other-results}

In this section, we consider the problem of reconstruction of harmonic tensors for which an equivariant reconstruction into second-order tensors is not possible, namely the harmonic tensors within the trigonal and the tetragonal classes. For these two classes, we show that such a reconstruction is possible up to a \emph{fourth-order harmonic tensor with cubic symmetry}. All the closed $\SO(3)$ subgroups used in this section are defined in \Cref{sec:symmetry-classes}.

\subsection{The tetragonal class}
\label{subsec:tetragonal}

An harmonic tensor $\bH\in \HH^{4}(\RR^{3})$ is tetragonal if and only if there exists $g\in \SO(3)$ such that $\bH = g\star \bH_{0}$ where $\bH_{0}$ has the following normal matrix form
\begin{equation}\label{eq:normal-form-Tetra}
  \underline{\bH_{0}} = \left(
  \begin{array}{cccccc}
      3 \delta-\sigma & \delta+\sigma   & -4 \delta & 0         & 0         & 0               \\
      \delta+\sigma   & 3 \delta-\sigma & -4 \delta & 0         & 0         & 0               \\
      -4 \delta       & -4 \delta       & 8 \delta  & 0         & 0         & 0               \\
      0               & 0               & 0         & -8 \delta & 0         & 0               \\
      0               & 0               & 0         & 0         & -8 \delta & 0               \\
      0               & 0               & 0         & 0         & 0         & 2\delta+2\sigma \\
    \end{array}
  \right)
\end{equation}
where $\sigma^{2}-25\delta^{2}\neq 0$ and $\sigma\neq 0$.

\begin{rem}
  Note that this normal form is however not unique; changing $\sigma$ to $-\sigma$ provides an alternative normal form.
  Nevertheless, the choice $\sigma>0$ allows to fix this ambiguity. Note also that
  \begin{equation*}
    \bH_{0}(-\sigma,\delta) = r \star \bH_{0}(\sigma,\delta),
  \end{equation*}
  where $r=R\left(\be_{3},\frac{\pi}{4}\right)$ is the rotation by angle $\frac{\pi}{4}$ around the $(Oz)$ axis.
\end{rem}

For tetragonal harmonic tensors, it was shown in~\cite[Section 5.4]{AKP2014} that the following polynomial invariants do not vanish
\begin{equation}\label{eq:Inv-K4-K10}
  K_{4}  :=3J_{4}-J_{2}^{2}> 0
  \qquad
  K_{10}  :=2 J_{2}K_{4}^{2}-35J_{5}^{2}  > 0
\end{equation}
and that $\delta$ and $\sigma^{2}$ are rational invariants, given by
\begin{equation}\label{eq:coeff-tetragonal-delta}
  \delta=\frac{1}{4}\frac{J_{5}}{K_{4}},
  \qquad
  \sigma^{2}=\frac{1}{8}\left(J_{2}-280\delta^{2}\right)=\frac{1}{16}\frac{K_{10}}{K_{4}^{2}}.
\end{equation}
The choice $\sigma>0$ in the normal form~\eqref{eq:normal-form-Tetra} allows to write $\sigma$ as follows:
\begin{equation}\label{eq:coeff-tetragonal-sigma}
  \sigma=\frac{\sqrt{K_{10}}}{4 K_{4}}.
\end{equation}

\begin{rem}
  Note that the condition $\sigma^{2}-25\delta^{2}= 0$ is equivalent to $L_{10}=0$ with
  \begin{equation*}
    L_{10}:=K_{10}-25J_{5}^{2},
  \end{equation*}
  and corresponds to the degeneracy when $\bH$ is at least in the cubic class. On the other hand, the condition $\sigma=0$ (with $L_{10}\ne0$) is equivalent to $K_{10}=0$ and corresponds to the degeneracy when $\bH$ is at least in the transverse isotropic class.
\end{rem}

\subsubsection{Geometric picture for the tetragonal class}

The normal form~\eqref{eq:normal-form-Tetra} corresponds to the subspace $Fix(\DD_{4})$, where given a subgroup $\Gamma$ of $\SO(3, \RR)$ the \emph{fixed point set} $Fix(\Gamma)$ is defined as
\begin{equation*}
  Fix(\Gamma):=\set{ \bH\in \HH^{4}(\RR^{3}),\quad g\star \bH=\bH,\quad \forall g\in \Gamma}.
\end{equation*}
Note now that there is only one subgroup in the conjugacy class $[\OO(2)]$ which contains $\DD_{4}$, it is $\OO(2)$ itself. Its fix point set, $Fix(\OO(2))$, is the one-dimensional subspace spanned by
\begin{equation} \label{eq:T0spanned}
  \underline{\pmb{\mathcal{T}}_{0}} := \left(
  \begin{array}{cccccc}
      3  & 1  & -4 & 0  & 0  & 0 \\
      1  & 3  & -4 & 0  & 0  & 0 \\
      -4 & -4 & 8  & 0  & 0  & 0 \\
      0  & 0  & 0  & -8 & 0  & 0 \\
      0  & 0  & 0  & 0  & -8 & 0 \\
      0  & 0  & 0  & 0  & 0  & 2
    \end{array}
  \right).
\end{equation}
However, there are two different subgroups $\octa_{1}$ and $\octa_{2}$ in the conjugacy class $[\octa]$ which contains $\DD_{4}$ and which are defined as follows. The first one, $\octa_{1}$, is the symmetry group of the cube with edges $(\pm 1;\pm 1;\pm 1)$.
The second one, $\octa_{2}$, is obtained by rotating $\octa_{1}$ around the $(Oz)$ axis by angle $\frac{\pi}{4}$.
In other words, $\octa_{2} := r\octa_{1}r^{-1}$, where $r$ is the rotation by angle $\frac{\pi}{4}$ around the $(Oz)$ axis. The fixed point sets $Fix(\octa_{1})$ and $Fix(\octa_{2})$ are one-dimensional subspaces spanned respectively by
\begin{equation*}
  \underline{\pmb{\mathcal{C}}_{0}^{1}} := \left(
  \begin{array}{cccccc}
      8  & -4 & -4 & 0  & 0  & 0  \\
      -4 & 8  & -4 & 0  & 0  & 0  \\
      -4 & -4 & 8  & 0  & 0  & 0  \\
      0  & 0  & 0  & -8 & 0  & 0  \\
      0  & 0  & 0  & 0  & -8 & 0  \\
      0  & 0  & 0  & 0  & 0  & -8 \\
    \end{array}
  \right)
\end{equation*}
and
\begin{equation*}
  \underline{\pmb{\mathcal{C}}_{0}^{2}} := \left(
  \begin{array}{cccccc}
      -2 & 6  & -4 & 0  & 0  & 0  \\
      6  & -2 & -4 & 0  & 0  & 0  \\
      -4 & -4 & 8  & 0  & 0  & 0  \\
      0  & 0  & 0  & -8 & 0  & 0  \\
      0  & 0  & 0  & 0  & -8 & 0  \\
      0  & 0  & 0  & 0  & 0  & 12 \\
    \end{array}
  \right).
\end{equation*}
Moreover, both $Fix(\octa_{1})$ and $Fix(\octa_{2})$ are supplementary subspaces of $Fix(\OO(2))$ in $Fix(\DD_{4})$, and thus we can write either
\begin{equation*}
  \bH_{0}(\sigma,\delta) = \frac{5\delta + \sigma}{5}\pmb{\mathcal{T}}_{0} - \frac{\sigma}{5}\pmb{\mathcal{C}}_{0}^{1},
\end{equation*}
or
\begin{equation*}
  \bH_{0}(\sigma,\delta) = \frac{5\delta - \sigma}{5}\pmb{\mathcal{T}}_{0} + \frac{\sigma}{5}\pmb{\mathcal{C}}_{0}^{2}.
\end{equation*}

Let $\bH$ be a tetragonal harmonic fourth order tensor. Then there exists $g \in \SO(3)$ (defined up to a right composition by an element of $\DD_{4}$) and a unique couple of real numbers $(\sigma, \delta)$ with $\sigma >0$ such that:
\begin{equation*}
  \bH = g \star \bH_{0}(\sigma,\delta).
\end{equation*}
Since $\DD_{4} \subset \OO(2)$, the following definitions
\begin{equation*}
  \pmb{\mathcal{T}}^{1}(\bH) := \frac{5\delta + \sigma}{5} g \star \pmb{\mathcal{T}}_{0}, \quad \pmb{\mathcal{T}}^{2}(\bH) := \frac{5\delta - \sigma}{5} g \star \pmb{\mathcal{T}}_{0},
\end{equation*}
do not depend on the rotation $g$, leading to well-defined covariant mappings $\pmb{\mathcal{T}}^{1}$, $\pmb{\mathcal{T}}^{2}$ from the tetragonal class in $\HH^{4}(\RR^{3})$ to the transversely isotropic classes. Similarly, since $\octa_{1}$ and $\octa_{2}$ contains $\DD_{4}$, the same is true for
\begin{equation*}
  \pmb{\mathcal{C}}^{1}(\bH) := - \frac{\sigma}{5} g \star \pmb{\mathcal{C}}_{0}^{1}, \quad \pmb{\mathcal{C}}^{2}(\bH) := \frac{\sigma}{5} g \star \pmb{\mathcal{C}}_{0}^{2}.
\end{equation*}

\begin{rem}
  Let $r=R\left(\be_{3},\frac{\pi}{4}\right)$ be the rotation by angle $\frac{\pi}{4}$ around the $(Oz)$ axis.
  We have then
  \begin{equation*}
    r \star \pmb{\mathcal{T}}_{0} = \pmb{\mathcal{T}}_{0}, \qquad r \star \pmb{\mathcal{C}}_{0}^{1} = \pmb{\mathcal{C}}_{0}^{2}.
  \end{equation*}
  Now if $\bH = g \star \bH_{0}(\sigma,\delta)$, we get
  \begin{equation*}
    (grg^{-1}) \star \pmb{\mathcal{T}}^{1}(\bH) = \pmb{\mathcal{T}}^{2}(\bH), \qquad (grg^{-1}) \star \pmb{\mathcal{C}}^{1}(\bH) = \pmb{\mathcal{C}}^{2}(\bH).
  \end{equation*}
\end{rem}

\subsubsection{Reconstruction theorem for the tetragonal class}

\begin{thm}\label{thm:D4}
  For any tetragonal fourth-order harmonic tensor $\bH$, we have
  \begin{equation*}
    \bH = \pmb{\mathcal{T}}^{k}(\bH) + \pmb{\mathcal{C}}^{k}(\bH), \qquad k = 1,2
  \end{equation*}
  where
  \begin{equation*}
    \pmb{\mathcal T}^{k}(\bH)=\frac{7}{16}
    \frac{(5\delta+(-1)^{k+1}\sigma) }{(25 {\delta^{2} }-\sigma^{2})^{2} }\,\mathbf{d}_{2}'(\bH)\ast \mathbf{d}_{2}'(\bH)
  \end{equation*}
  is a transversely isotropic covariant, and
  \begin{equation*}
    \pmb{\mathcal{C}}^{k}(\bH) =   \left(1-\frac{14\delta}{5 \delta - (-1)^{k+1}\sigma}\right)\bH + \frac{7}{2(5\delta - (-1)^{k+1}\sigma)} (\bH^{2})_{0},
  \end{equation*}
  is a cubic covariant.
\end{thm}

\begin{proof}
  All the formulas can be checked on the normal form~\eqref{eq:normal-form-Tetra}, for which
  $\bd_{2}'(\bH)=4 \left(25 \delta^{2}-\sigma ^{2}\right)\, (\vec e_{3} \otimes \vec e_{3})'$.
  Since the formulas are covariant, this is enough to achieve the proof.
\end{proof}

\begin{rem}
  In theorem~\ref{thm:D4}, each formula, i.e. $k=1$ or $k=2$, is self-sufficient. It is equivariant and applies in any basis to any tetragonal harmonic fourth order tensor $\bH$. The two reconstruction formulas correspond to the choice of a frame and therefore of a cube, rotated by the angle $\frac{\pi}{4}$ or not around $(Oz)$. They use the invariants $\sigma>0$ and $\delta$ defined in \eqref{eq:coeff-tetragonal-delta} and \eqref{eq:coeff-tetragonal-sigma}.
\end{rem}

The decomposition in theorem~\ref{thm:D4} can be rewritten in terms of the invariants $J_{5}$, $K_{4}$, $K_{10}$ and $L_{10}$, rather than $\sigma,\delta$ by using their rational expressions~\eqref{eq:coeff-tetragonal-delta} and~\eqref{eq:coeff-tetragonal-sigma}.

\begin{cor}\label{cor:D4}
  For any tetragonal fourth-order harmonic tensor $\bH$, we have
  \begin{equation*}
    \bH = \frac{28 K_{4}^{3} (5 J_{5}+\sqrt{K_{10}})}{L^{2}_{10}}\, \bd'_{2}(\bH) \ast \bd'_{2}(\bH) + \pmb{\mathcal{C}}(\bH)
  \end{equation*}
  where
  \begin{equation*}
    \pmb{\mathcal{C}}(\bH)=\left(1+\frac{14 J_{5}(5J_{5}+\sqrt{K_{10}})}{L_{10}}  \right)\;\bH
    -\frac{14 K_{4}(5J_{5}+\sqrt{K_{10}})}{L_{10}}\, (\bH^{2})_{0}
  \end{equation*}
  is a cubic fourth order covariant and
  \begin{align*}
    K_{4}  & =3J_{4}-J_{2}^{2}>0,
           &
    K_{10} & =2 J_{2}K_{4}^{2}-35J_{5}^{2}>0,
           &
    L_{10} & = K_{10}-25J_{5}^{2}\neq 0.
  \end{align*}
\end{cor}

\subsection{The trigonal class}
\label{subsec:trigonal}

An harmonic tensor $\bH\in \HH^{4}(\RR^{3})$ is trigonal if and only if there exists $g\in \SO(3)$ such that $\bH = g\star \bH_{0}$ where $\bH_{0}$ has the following normal matrix form
\begin{equation}\label{eq:normal-form-Trigonal}
  \underline{\bH_{0}} = \left(
  \begin{array}{cccccc}
      3 \delta        & \delta         & -4 \delta & -\sqrt{2}\sigma & 0         & 0        \\
      \delta          & 3 \delta       & -4 \delta & \sqrt{2}\sigma  & 0         & 0        \\
      -4 \delta       & -4 \delta      & 8 \delta  & 0               & 0         & 0        \\
      -\sqrt{2}\sigma & \sqrt{2}\sigma & 0         & -8 \delta       & 0         & 0        \\
      0               & 0              & 0         & 0               & -8 \delta & -2\sigma \\
      0               & 0              & 0         & 0               & -2\sigma  & 2\delta
    \end{array}
  \right)
\end{equation}
where $\sigma^{2}-50\delta^{2}\neq 0$ and $\sigma\neq 0$.

\begin{rem}
  Note that this normal form is however not unique. Changing $\sigma$ to $-\sigma$ provides an alternative normal form.
  Nevertheless, the choice $\sigma>0$ allows to fix this ambiguity. Note also that
  \begin{equation*}
    \bH_{0}(-\sigma,\delta)= r_t \star \bH_{0}(\sigma,\delta),
  \end{equation*}
  where $r_t=R\left(\be_{3},\frac{\pi}{3}\right)$ is the rotation by angle $\frac{\pi}{3}$ around the $(Oz)$ axis.
\end{rem}

For trigonal harmonic tensors, it was shown in~\cite[Section 5.3]{AKP2014} that the polynomial invariants $K_{4}$ and $K_{10}$, defined by \eqref{eq:Inv-K4-K10}, are strictly positive and that $\delta$, $\sigma^{2}$ are rational invariants, given by
\begin{equation}\label{eq:coeff-trigonal}
  \delta=\frac{1}{4}\frac{J_{5}}{K_{4}},\quad \sigma^{2}=\frac{1}{16}\left(J_{2}-280\delta^{2}\right)=\frac{1}{32}\frac{K_{10}}{K_{4}^{2}}.
\end{equation}

The choice $\sigma>0$ in the normal form~\eqref{eq:normal-form-Trigonal} allows to write $\sigma$ as follows:
\begin{equation}\label{eq:coeff-trigonal-sigma}
  \sigma=\frac{1}{4\sqrt{2}}\frac{\sqrt{K_{10}}}{K_{4}}.
\end{equation}

\begin{rem}
  Note that the condition $\sigma^{2}-50\delta^{2}= 0$ is equivalent to $M_{10}=0$ with
  \begin{equation*}
    M_{10}:=K_{10}-100J_{5}^{2},
  \end{equation*}
  and corresponds to the degeneracy case when $\bH$ has at least cubic symmetry.
  On the other hand, with $M_{10}\ne 0$, condition $\sigma=0$ is equivalent to $K_{10}=0$ and corresponds to the degeneracy case when $\bH$ has at least transverse isotropic symmetry.
\end{rem}

\subsubsection{Geometric picture for the trigonal class}

The geometric picture is similar to the tetragonal case: $\DD_{3}$ is contained in only one subgroup in the conjugacy class $[\OO(2)]$ and two different subgroups $\widetilde{\octa}_{1}$, $\widetilde{\octa}_{2}$ in the conjugacy class $[\octa]$, which are defined as follows.
\begin{equation*}
  \widetilde{\octa}_{1} := r_{3}\mathbb{O}_{1}r_{3}^{-1}, \qquad \widetilde{\octa}_{2} := r_t\widetilde{\octa}_{1}r_{t}^{-1}
\end{equation*}
where $\octa_{1}$ is the symmetry group of the cube with edges $(\pm 1;\pm 1;\pm 1)$ (already defined in section~\ref{subsec:tetragonal}) and the rotations introduced are
\begin{align*}
  r_{3} & := R\left(\be_{3},\frac{\pi}{4}\right) \circ R\left(\be_{1}-\be_{2},\arccos\left(\frac{1}{\sqrt{3}}\right)\right),
        &
  r_t   & : = R\left(\be_{3},\frac{\pi}{3}\right).
\end{align*}

The fixed point sets $Fix(\widetilde{\octa}_{1})$ and $Fix(\widetilde{\octa}_{2})$ are one-dimensional subspaces spanned respectively by
\begin{equation*}
  \underline{\widetilde{\pmb{\mathcal{C}}}_{0}^{1}} :=
  \left(
  \begin{array}{cccccc}
      3   & 1  & -4 & -10 & 0             & 0             \\
      1   & 3  & -4 & 10  & 0             & 0             \\
      -4  & -4 & 8  & 0   & 0             & 0             \\
      -10 & 10 & 0  & -8  & 0             & 0             \\
      0   & 0  & 0  & 0   & -8            & -10\,\sqrt{2} \\
      0   & 0  & 0  & 0   & -10\,\sqrt{2} & 2
    \end{array}
  \right)
\end{equation*}
and
\begin{equation*}
  \underline{\widetilde{\pmb{\mathcal{C}}}_{0}^{2}} :=
  \left(
  \begin{array}{cccccc}
      3  & 1   & -4 & 10  & 0            & 0            \\
      1  & 3   & -4 & -10 & 0            & 0            \\
      -4 & -4  & 8  & 0   & 0            & 0            \\
      10 & -10 & 0  & -8  & 0            & 0            \\
      0  & 0   & 0  & 0   & -8           & 10\,\sqrt{2} \\
      0  & 0   & 0  & 0   & 10\,\sqrt{2} & 2
    \end{array}
  \right).
\end{equation*}

The normal form~\eqref{eq:normal-form-Trigonal} can be decomposed as the sum of a transversely isotropic harmonic tensor and a cubic harmonic tensor, either as
\begin{equation*}
  \bH_{0}(\sigma,\delta) = \frac{10\delta - \sigma\sqrt{2}}{10}\pmb{\mathcal{T}}_{0} + \frac{\sigma\sqrt{2}}{10}\widetilde{\pmb{\mathcal{C}}}_{0}^{1},
\end{equation*}
or as
\begin{equation*}
  \bH_{0}(\sigma,\delta) = \frac{10\delta + \sigma\sqrt{2}}{10}\pmb{\mathcal{T}}_{0} - \frac{\sigma\sqrt{2}}{10}\widetilde{\pmb{\mathcal{C}}}_{0}^{2},
\end{equation*}
where $\pmb{\mathcal{T}}_{0}$ has been defined in~\eqref{eq:T0spanned}.

As in the tetragonal case, we can give coherent definitions of transversely isotropic parts

\begin{equation*}
  \widetilde{\pmb{\mathcal{T}}}^{1}(\bH) := \frac{10\delta - \sigma\sqrt{2}}{10}g\star \pmb{\mathcal{T}}_{0}, \quad \widetilde{\pmb{\mathcal{T}}}^{2}(\bH) :=  \frac{10\delta + \sigma\sqrt{2}}{10} g \star \pmb{\mathcal{T}}_{0},
\end{equation*}
and cubic parts
\begin{equation*}
  \widetilde{\pmb{\mathcal{C}}}^{1}(\bH) := \frac{\sigma\sqrt{2}}{10}g\star \widetilde{\pmb{\mathcal{C}}}_{0}^{1}, \quad \widetilde{\pmb{\mathcal{C}}}^{2}(\bH) := - \frac{\sigma\sqrt{2}}{10}g\star \widetilde{\pmb{\mathcal{C}}}_{0}^{2},
\end{equation*}

\subsubsection{Reconstruction formulas for the trigonal class}

\begin{thm}\label{thm:D3}
  For any trigonal fourth-order harmonic tensor $\bH$, we have
  \begin{equation*}
    \bH = \widetilde{\pmb{\mathcal{T}}}^{k}(\bH) + \widetilde{\pmb{\mathcal{C}}}^{k}(\bH), \qquad k = 1,2,
  \end{equation*}
  where
  \begin{equation*}
    \widetilde{\pmb{\mathcal T}}^k(\bH)= \frac{7 \left(10 \delta-(-1)^{k+1}\sigma\sqrt{2} \right)}{8 \left(50 \delta^{2}-\sigma ^{2}\right)^{2}}
    \,\mathbf{d}_{2}'(\bH)\ast \mathbf{d}_{2}'(\bH),
  \end{equation*}
  is a transversely isotropic covariant and
  \begin{equation*}
    \widetilde{\pmb{\mathcal{C}}}^{k}(\bH) = \left(1-\frac{7\delta}{10\delta+ (-1)^{k+1}\sigma\sqrt{2}}\right)\bH - \frac{7}{6(10\delta + (-1)^{k+1}\sigma\sqrt{2})} (\bH^{2})_{0}
  \end{equation*}
  is a cubic covariant.
\end{thm}

\begin{proof}
  All the formulas can be checked on the normal form~\eqref{eq:normal-form-Trigonal}, for which
  $\bd_{2}'(\bH)=2 \left(50 \delta^{2}-\sigma ^{2}\right)\, (\vec e_{3} \otimes \vec e_{3})'$.
  Since the formulas are covariant, this is enough to achieve the proof.
\end{proof}

\begin{rem}
  In theorem~\ref{thm:D3}, each formula, i.e. $k=1$ or $k=2$, is self-sufficient. It is equivariant and applies to any trigonal fourth-order harmonic tensor $\bH$ in any basis. The two reconstruction formulas correspond to the choice of a frame and therefore of a cube, rotated by the angle $\frac{\pi}{3}$ or not around $(Oz)$. They use the invariants $\sigma>0$ and $\delta$ defined in Eq. \eqref{eq:coeff-trigonal} and \eqref{eq:coeff-trigonal-sigma}.
\end{rem}

The decomposition in theorem~\ref{thm:D4} can be rewritten in terms of the invariants $J_{5}$, $K_{4}$, $K_{10}$ and $M_{10}$, rather than $\sigma,\delta$ by using their rational expressions~\eqref{eq:coeff-trigonal} and~\eqref{eq:coeff-trigonal-sigma}.

\begin{cor}\label{cor:D3}
  For any trigonal fourth-order harmonic tensor $\bH$, we have
  \begin{equation*}
    \bH = \frac{224 K_{4}^{3} (10 J_{5}+ \sqrt{K_{10}})}{M^{2}_{10}}\, \bd'_{2}(\bH) \ast \bd'_{2}(\bH) + \widetilde{\pmb{\mathcal{C}}}(\bH)
  \end{equation*}
  where
  \begin{equation*}
    \widetilde{ \pmb{\mathcal{C}}}(\bH)= \left(1+\frac{7 J_{5}(10J_{5}+\sqrt{K_{10}})}{M_{10}} \right) \bH
    +\frac{14K_{4}(10J_{5}+\sqrt{K_{10}})}{3M_{10}}\, (\bH^{2})_{0}
  \end{equation*}
  is a cubic fourth order covariant, and
  \begin{align*}
    K_{4}  & =3J_{4}-J_{2}^{2}>0,
           &
    K_{10} & =2 J_{2}K_{4}^{2}-35J_{5}^{2}>0,
           &
    M_{10} & = K_{10}-100 J_{5}^{2}\neq 0.
  \end{align*}
\end{cor}

\section{Conclusion}
\label{sec:conclusion}

In this article, we have defined the \emph{harmonic projection} $(\bT)_{0}\in \HH^{n}(\RR^{3})$ of a totally symmetric $n$-th order tensor $\bT$, and the \emph{harmonic product}
\begin{equation*}
  \bH_{1}\ast \bH_{2} := (\bH_{1}\odot \bH_{2})_{0} \in \HH^{n_{1}+n_{2}}(\RR^{3})
\end{equation*}
of two harmonic tensors $\bH_{1} \in \HH^{n_{1}}(\RR^{3})$, $\bH_{2} \in \HH^{n_{2}}(\RR^{3})$.

Using Sylvester's theorem~\ref{thm:Sylvester}, we have shown that the fourth order harmonic part $\bH$ of an elasticity tensor $\tq E$ can be expressed as the harmonic product $\bH=\bh_{1}\ast\bh_{2}$ of two second order harmonic tensors $\bh_{1}, \bh_{2}$ (\emph{i.e.} symmetric deviatoric). This decomposition is independent of any reference frame, even for triclinic materials. Nevertheless, such a factorization is non-unique and not really constructive. Moreover, a \emph{globally defined} solution $(\bh_{1}(\bH),\bh_{2}(\bH))$ can never be \emph{covariant}.

We have therefore formulated an \emph{equivariant reconstruction} problem of a tensor by means of \emph{lower order covariants}, restricted to a given symmetry class. Finally, we have obtained explicit formulas for the reconstruction of the fourth order harmonic part $\bH$ of an elasticity tensor:
\begin{itemize}
  \item for the \emph{transversely isotropic} and \emph{orthotropic} symmetry classes, by means of second order covariants (theorems~\ref{thm:O2} and theorem~\ref{thm:D2}). Besides, necessary and sufficient conditions for such tensors to be perfect harmonic squares, which means $\bH=\bh\ast \bh$, have also been obtained (Remark~\ref{rem:O2-harmonic-square} and theorem~\ref{thm:Orth_Square});
  \item for the \emph{tetragonal} and \emph{trigonal} symmetry classes, by means of second order covariants and a fourth order cubic covariant (theorem~\ref{thm:D4} and theorem~\ref{thm:D3}).
\end{itemize}
These reconstruction formulas are only valid inside each symmetry class, this is illustrated by the fact that they contain denominators which vanish when a degeneracy into a higher symmetry class occurs.

\appendix

\section{Binary forms}
\label{sec:Annexe_binary-forms}

We have already described two models for the irreducible representations of the rotation group $\SO(3)$: the space of \emph{harmonic tensors} $\HH^{n}(\RR^{3})$ and the space of harmonic polynomials $\Hn{n}(\RR^{3})$. In this section, we shall derive a third model: the space of \emph{binary forms} $\SRn{2n}$, whose construction is slightly more cumbersome.

To start, let us recall that there is a well-known relation between the rotation group $\SO(3)$ and the group of \emph{special unitary complex matrices}
\begin{equation*}
  \SU(2) := \set{\gamma \in \mathrm{M}_{2}(\CC); \; \bar{\gamma}^{t}\gamma = \Idd ,\, \det \gamma = 1}.
\end{equation*}
Let
\begin{equation*}
  \xx := (x,y,z) \mapsto M(\xx) = \left(
  \begin{array}{cc}
      -z   & x+iy \\
      x-iy & z    \\
    \end{array}
  \right)
\end{equation*}
be a linear mapping from $\RR^{3}$ to the space of traceless, hermitian matrices of order $2$ ($\sigma_{x} = M(\be_{1})$, $\sigma_{y} = M(\be_{2})$ and $\sigma_{z} = M(\be_{3})$ are the famous \emph{Pauli matrices}). The group $\SU(2)$ acts on this space by conjugacy
\begin{equation*}
  \Ad_{\gamma} : M \mapsto \gamma M \gamma^{-1}, \qquad \gamma \in \SU(2),
\end{equation*}
and preserves the quadratic form
\begin{equation*}
  \det M = -(x^{2} + y^{2} + z^{2}).
\end{equation*}
It can be checked, moreover, that $\det \Ad_{\gamma} = 1$ for all $\gamma \in \SU(2)$. Therefore, we deduce a group morphism
\begin{equation*}
  \pi : \gamma \mapsto \Ad_{\gamma}, \qquad \SU(2) \to \SO(3),
\end{equation*}
{whose kernel consists of two element $\Idd$ and $-\Idd$. $\SU(2)$ is a \emph{double cover} of the rotation group $\SO(3)$: to each rotation $g\in \SO(3)$, it corresponds exactly two elements in $\SU(2)$, namely $\gamma$ and $-\gamma$, such that $\pi(\gamma) = g$.}

\begin{ex}
  For
  \begin{equation*}
    \gamma=\begin{pmatrix}
      \text{e}^{i\theta} & 0                   \\
      0                  & \text{e}^{-i\theta}
    \end{pmatrix}
  \end{equation*}
  we get
  \begin{equation*}
    \pi(\gamma)=
    \begin{pmatrix}
      \cos(2\theta) & -\sin(2\theta) & 0 \\
      \sin(2\theta) & \cos(2\theta)  & 0 \\
      0             & 0              & 1
    \end{pmatrix}
  \end{equation*}
  in the basis $(\sigma_{x}, \sigma_{y}, \sigma_{z})$.
\end{ex}

The group $\SU(2)$ acts naturally on $\CC^{2}$ and more generally on complex polynomials in two variables $(u,v)$ by the rule
\begin{equation*}
  (\gamma \star \ff)(u,v) := \ff(\gamma^{-1} \star (u,v)), \qquad \gamma \in \SU(2).
\end{equation*}

Let $\Sn{n}(\CC^{2})$ be the space of complex, homogeneous polynomials of degree $n$ in two variables. An element of $\Sn{n}(\CC^{2})$
\begin{equation*}
  \ff(\bxi) := \sum_{k=0}^n a_{k} u^{k}v^{n-k},\quad \bxi := (u,v)\in \CC^{2}
\end{equation*}
is called a \emph{binary form} of degree $n$.

The main observation, due to Cartan (see~\cite{OKA2017} for the details), is that there is an isomorphism between the space $\Sn{2n}(\CC^{2})$ of binary forms of degree $2n$ and the space $\Hn{n}(\CC^{3})$ of complex, harmonic polynomials of degree $n$ in three variables $x,y,z$. This linear isomorphism $\psi: \Hn{n}(\CC^{3}) \to \Sn{2n}(\CC^{2})$ is given explicitly by
\begin{equation}\label{eq:Cartan-map}
  (\psi (\rh))(u,v) := \rh \left( \frac{u^{2}-v^{2}}{2}, \frac{u^{2}+v^{2}}{2i}, uv \right), \qquad \rh \in \Hn{n}(\CC^{3}).
\end{equation}

\begin{rem}[Real harmonic polynomials]
  Under this isomorphism, the space $\Hn{n}(\RR^{3})$ of \emph{real} harmonic polynomials corresponds to the \emph{real subspace} $\SRn{2n} \subset \Sn{2n}(\CC^{2})$ of binary forms $\ff$ which satisfy
  \begin{equation}\label{eq:functional-characterization}
    \overline{\ff}(-v,u)=(-1)^{n}\ff(u,v).
  \end{equation}
  These binary forms
  \begin{equation*}\label{eq:genericbinaryf}
    \ff(u,v) := \sum_{k=0}^{2n} a_{k}u^{k}v^{2n-k}
  \end{equation*}
  are also characterized by the following property:
  \begin{equation}\label{eq:coefficients-characterization}
    a_{2n-k} = (-1)^{n-k} \overline{a_{k}}, \qquad 0 \le k \le n.
  \end{equation}
\end{rem}

Moreover, it was shown in~\cite{OKA2017}, that the \emph{Cartan isomorphism} $\psi$ commutes with the action of the complex group $\SL(2,\CC)$ (the set of two by two complex matrices with determinant $1$):
\begin{equation*}
  \psi (\pi(\gamma) \star \rh) = \gamma \star \psi(\rh), \qquad \gamma \in \SL(2,\CC),\quad \rh \in \Hn{n}(\CC^{3}).
\end{equation*}
The space $\SRn{2n}$ is invariant under the action of the subgroup $\SU(2)$ of $\SL(2,\CC)$ and is moreover \emph{irreducible} for this action. Since $-\Idd$ acts as the identity on even order binary forms, the action of $\SU(2)$ on $\SRn{2n}$ reduces to an action of the rotation group $\SO(3)$ on $\SRn{2n}$. The isomorphism $\psi$ induces therefore an \emph{$\SO(3)$-equivariant isomorphism} between $\SRn{2n}$ and $\Hn{n}(\RR^{3})$.

The inverse $\psi^{-1}$ can be defined explicitly as followed. Let
\begin{equation*}
  \ff(u,v) := \sum_{k=0}^{2n} a_{k}u^{k}v^{2n-k}.
\end{equation*}
For each $k$ make the substitution
\begin{equation*}
  u^{k}v^{2n-k} \rightarrow
  \left\{ \begin{array}{cc}
    z^{k}(-x+iy)^{n-k},   & \text{if } 0 \le k \le n   \\
    z^{2n-k}(x+iy)^{k-n}, & \text{if } n \le k \le 2n.
  \end{array}
  \right.
\end{equation*}
We obtain this way a homogeneous polynomial in three variables $\rp$ of degree $n$ such that
\begin{equation*}
  \ff(u,v) = \rp \left(\frac{u^{2}-v^{2}}{2},\frac{u^{2}+v^{2}}{2i},uv\right).
\end{equation*}
Now, let $\rh = (\rp)_{0} $ be the $n$-th order harmonic component of $\rp$ in the harmonic decomposition~\eqref{eq:symmetric-harmonic-decomposition} of $\rp$. Then $\psi^{-1}(\ff) := (\rp)_{0} = \rh$.

\begin{rem}\label{rem:multiplicative-property}
  Note that if $\ff_{1} \in \SRn{2n_{1}}$ and $\ff_{2} \in \SRn{2n_{2}}$, then $\ff_{1}\ff_{2} \in \SRn{2(n_{1}+n_{2})}$ and
  \begin{equation*}
    \psi^{-1}(\ff_{1}\ff_{2}) = (\psi^{-1}(\ff_{1}) \psi^{-1}(\ff_{2}))_{0}
  \end{equation*}
\end{rem}

\begin{ex}[Order-$1$ harmonic tensors]\label{ex:order1-harmonic-tensors}
  An order-$1$ harmonic tensor $\bh$ on $\RR^{3}$ is just a linear form $\rh$ on $\RR^{3}$, which can be written as
  \begin{equation*}
    \rh(\xx) = \ww \cdot \xx = w_{i} \, x_{i},
  \end{equation*}
  with $\vec w=(w_{1}, w_{2}, w_{3})$. The corresponding binary form $\ff=\psi(\rh)\in \SRn{2}$ is given by
  \begin{equation*}
    \ff(u,v) = \overline{a_{0}}u^{2} + a_{1}uv - a_{0}v^{2},
  \end{equation*}
  where
  \begin{equation*}
    a_{0} = \frac{1}{2} \left(w_{1} +iw_{2}\right), \quad a_{1}= w_{3}.
  \end{equation*}
  Conversely, given $\ff\in \SRn{2}$, the corresponding real $3$-vector $\ww$ is given by
  \begin{equation*}
    w_{1} = a_{0} + \overline{a_{0}}, \quad w_{2} = -i\left(a_{0} - \overline{a_{0}}\right), \quad w_{3} = a_{1}.
  \end{equation*}
\end{ex}

\begin{rem}\label{rem:stereographic-projection}
  For order-$1$ tensors, there is a closed relation between the Cartan-map~\eqref{eq:Cartan-map} and the stereographic projection
  \begin{equation*}
    \tau:\ww^{0} = (w^{0}_{1}, w^{0}_{2}, w^{0}_{3}) \mapsto \lambda: = \frac{w^{0}_{1}}{1 - w^{0}_{3}} + i\frac{w^{0}_{2}}{1-w^{0}_{3}}, \quad S^{2}\setminus \set{\nn}\longrightarrow \CC
  \end{equation*}
  where $S^{2}$ is the unit sphere and $\vec n=(0,0,1)$. The stereographic projection can be extended to a bijection $\tau :S^{2} \to \CC \cup \set{\infty}$ by setting $\tau(\nn) = \infty$. Its inverse, $\tau^{-1}$, is given by
  \begin{equation*}
    w^{0}_{1} = \frac{2}{\mid \lambda\mid^{2}+1}\Re(\lambda),\quad w^{0}_{2} = \frac{2}{\mid \lambda\mid^{2}+1}\Im(\lambda),\quad w^{0}_{3} = \frac{\mid \lambda\mid^{2}-1}{\mid \lambda\mid^{2}+1}, \quad \lambda\in \CC.
  \end{equation*}
  Now let $\ww \ne (0,0,0)$. Then we can write $\ww = t \ww^{0}$, where $t = \norm{\ww}$ and $\norm{\ww^{0}} = 1$. If moreover $\ww^{0} \ne (0,0,1)$, we can write $\ff=\psi(\rh)$ as
  \begin{equation}\label{eq:stereographic-projection-representation}
    \psi(\rh) (u,v) = t\alpha( u +\lambda v)(\overline{\lambda}u-v),\quad \alpha:=\frac{1}{\mid \lambda\mid^{2}+1},\quad \lambda=\tau(\ww^{0}).
  \end{equation}
\end{rem}

\begin{ex}[Order-$2$ harmonic tensors]\label{ex:order2-harmonic-tensors}
  Let $\bh=(h_{ij})$ be a second order real harmonic tensor and $\rh$ be the corresponding real harmonic polynomial, then the corresponding binary form $\ff=\psi(\rh)\in \SRn{4}$ is given by
  \begin{equation*}
    \ff(u,v) = \overline{a_{0}}u^{4} - \overline{a_{1}}u^{3}v + a_{2}u^{2}v^{2} + a_{1}uv^{3} + a_{0}v^{4},
  \end{equation*}
  where
  \begin{equation*}
    a_{0}=\frac{1}{4}h_{11}-\frac{1}{4}h_{22}+\frac{1}{2} i h_{12},\quad
    a_{1}=-h_{13}- i h_{23},\quad
    a_{2}=\frac{3}{2} h_{33}
  \end{equation*}
  Conversely, given $\ff\in \SRn{4}$, the corresponding harmonic tensor $\bh$ is given by
  \begin{equation}\label{eq:B4-to-H2}
    \bh:=\left(\begin{array}{ccc}
        a_{0} + \overline{a_{0}} - \frac{1}{3}a_{2}        & -i\left(a_{0} - \overline{a_{0}}\right)            &
        -\frac{1}{2} \left(a_{1} + \overline{a_{1}}\right)                                                                                                           \\
        -i\left(a_{0} - \overline{a_{0}}\right)            & -a_{0} - \overline{a_{0}} - \frac{1}{3}a_{2}       & \frac{1}{2} i\left(a_{1} - \overline{a_{1}}\right) \\
        -\frac{1}{2} \left(a_{1} + \overline{a_{1}}\right) & \frac{1}{2} i\left(a_{1} - \overline{a_{1}}\right) & \frac{2}{3}a_{2}
      \end{array}
    \right)
  \end{equation}
\end{ex}

\begin{ex}[Order-$4$ harmonic tensors]\label{ex:order4-harmonic-tensors}
  Let $\td H=(H_{ijkl})$ be a fourth-order real harmonic tensor and $\rh$ be the corresponding real harmonic polynomial, then the corresponding binary form $\ff=\psi(\rh)\in \SRn{8}$ is given by
  \begin{multline*}
    \ff(u,v) = \overline{a_{0}}u^{8} - \overline{a_{1}}u^{7}v + \overline{a_{2}}u^{6}v^{2} - \overline{a_{3}}u^{5}v^{3} + a_{4}u^{4}v^{4} \\
    + a_{3}u^{3}v^{5}+a_{2}u^{2}v^{6}+a_{1}uv^{7}+a_{0}v^{8},
  \end{multline*}
  where the independent components are
  \begin{align*}
    a_{0} & = -\frac{1}{16} (8 H_{1122}+H_{1133}+H_{2233})+i \frac{1}{4} (2  H_{1112}+  H_{1233}), \\
    a_{1} & = \frac{1}{2} (4 H_{1223}+H_{1333})+i \frac{1}{2} (H_{2223}-3  H_{1123}),              \\
    a_{2} & = \frac{7}{4} (H_{1133}-H_{2233})+i \frac{7}{2}  H_{1233},                             \\
    a_{3} & = -\frac{7}{2}  H_{1333}+i \frac{7}{2}  (H_{1123}+H_{2223}),                           \\
    a_{4} & = -\frac{35}{8} (H_{1133}+H_{2233})
  \end{align*}
  Conversely, given $\ff\in \SRn{8}$, the corresponding harmonic tensor in $\HH^4(\RR^{3})$ is given by:
  \begin{align*}\label{eq:B8-to-H4}
    H_{1111} & = - H_{1122} - H_{1133},                                                                      & H_{1112}                                                                                     & = \frac{i}{4}\Big(4\overline{a_{0}}-4a_{0}-\frac{2}{7}\overline{a_{2}}+\frac{2}{7}a_{2}\Big), \\
    H_{1113} & = - H_{1223} - H_{1333},
             & H_{1122}                                                                                      & = -a_{0}-\overline{a_{0}}+\frac{1}{35}a_{4},                                                                                                                                                 \\
    H_{1123} & = \frac{i}{4} \Big(a_{1}-\overline{a_{1}}+\frac{1}{7}\overline{a_{3}}-\frac{1}{7}a_{3} \Big), & H_{1133}                                                                                     & = \frac{1}{7}\overline{a_{2}}+\frac{1}{7}a_{2}-\frac{4}{35}a_{4},                             \\
    H_{1222} & = -H_{1112}-H_{1233},
             & H_{1223}                                                                                      & = \frac{1}{4}\Big( \overline{a_{1}}+a_{1}+\frac{1}{7}\overline{a_{3}}+\frac{1}{7}a_{3}\Big),                                                                                                 \\
    H_{1233} & = \frac{1}{7}i\left(\overline{a_{2}}-a_{2}\right),                                            & H_{1333}                                                                                     & = -\frac{1}{7}\left( a_{3}+\overline{a_{3}}\right),                                           \\
    H_{2222} & = -H_{1122}-H_{2233},
             & H_{2223}                                                                                      & = \frac{i}{4}\Big(\overline{a_{1}}-a_{1}+\frac{3}{7}\overline{a_{3}}-\frac{3}{7}a_{3} \Big),                                                                                                 \\
    H_{2233} & = -\frac{1}{7}a_{2}-\frac{1}{7}\overline{a_{2}}-\frac{4}{35}a_{4},                            & H_{2333}                                                                                     &
    = - H_{1123} - H_{2223},                                                                                                                                                                                                                                                                                \\
    H_{3333} & = -H_{1133}-H_{2233}.                                                                         &                                                                                              &
  \end{align*}
\end{ex}

\section{Proof of Sylvester's theorem}
\label{sec:Annexe_Sylvester-theorem-proof}

The key point to prove Sylvester's theorem~\ref{thm:Sylvester} is a \emph{root characterization} of binary forms in $\SRn{2n}$.

\begin{lem}\label{lem:root-characterization}
  Let $\ff\in \Sn{2n}(\CC^{2})$. Then $\ff \in \SRn{2n}$ if and only if it can be written as
  \begin{equation*}
    \ff(u,v) = \alpha u^{r}v^{r}\prod_{i=1}^{n-r} \left( u-\lambda_{i}v\right) \left( \overline{\lambda_{i}}u+v\right),
  \end{equation*}
  where $\lambda_{1},\dotsc,\lambda_{n-r}\in \CC^{*}$ and $\alpha \in \RR$.
\end{lem}

\begin{proof}
  Let $\ff \in \SRn{2n}$ which does not vanish identically. Due to~\eqref{eq:coefficients-characterization}, we deduce that for each $0 \le k \le n$, $a_{2n-k} = 0$ if and only if $a_{k}=0$. Thus, there exists a non-negative integer $r\leq n$ such that
  \begin{equation*}
    \ff(u,v) = u^{r}v^{r} \bg(u,v),
  \end{equation*}
  where $\bg \in \SRn{2(n-r)}$ with leading term $b_{2(n-r)} = (-1)^{n-r}\overline{b_{0}} \ne 0$. Let
  \begin{equation*}
    \rp(t) := \bg(t,1).
  \end{equation*}
  Then $\rp$ is a polynomial of degree $2(n-r)$ with no vanishing root. Due to~\eqref{eq:functional-characterization}, we deduce that if $\lambda$ is a root of $\rp$, then $-1/\bar{\lambda}$ is another root of $\rp$. Thus
  \begin{equation*}
    \rp(t) = b_{2(n-r)} \prod_{i=1}^{n-r} \left( t-\lambda_{i} \right) \left( t+\frac{1}{\overline{\lambda_{i}}}\right),
  \end{equation*}
  and we get
  \begin{equation*}
    \bg(u,v) = \frac{b_{2(n-r)}}{\prod \overline{\lambda_{i}}} \prod_{i=1}^{n-r} \left( u-\lambda_{i}v\right) \left( \overline{\lambda_{i}}u+v\right)
  \end{equation*}
  Now, using~\eqref{eq:functional-characterization}, we have
  \begin{equation*}
    \frac{\prod \lambda_{i}}{\prod \overline{\lambda_{i}}} = (-1)^{n-r} \frac{b_{0}}{b_{2(n-r)}} = \frac{\overline{b_{2(n-r)}}}{b_{2(n-r)}}
  \end{equation*}
  and thus
  \begin{equation}\label{eq:hpratique}
    \alpha := \frac{b_{2(n-r)}}{\prod \overline{\lambda_{i}}} = \frac{\overline{b_{2(n-r)}}}{\prod \lambda_{i}}
  \end{equation}
  is real. The converse can be immediately established using the functional characterization~\eqref{eq:functional-characterization} of $\SRn{2n}$.
\end{proof}

\begin{rem}\label{rem:about-non-unicity}
  This factorization is not unique. Permuting a couple of paired roots $\lambda_{i}, -1/\overline{\lambda_{i}}$, leads to the transformation
  \begin{equation*}
    \left( u-\lambda_{i}v\right) \left( \overline{\lambda_{i}}u+v\right) \to \frac{-1}{\abs{\lambda_{i}}^{2}} \left( u-\lambda_{i}v\right) \left( \overline{\lambda_{i}}u+v\right),
  \end{equation*}
  and may result in a different value for $\alpha$.
\end{rem}

We can now give the proof of Sylvester's theorem~\ref{thm:Sylvester} and of proposition~\ref{prop:fiber}.

\begin{proof}[Proof of theorem~\ref{thm:Sylvester}]
  Using the isomorphism between $\Hn{n}(\RR^{3})$ and $\SRn{2n}$ described in ~\Cref{sec:Annexe_binary-forms}, we set $\ff := \psi(\rp)$. Using Lemma~\ref{lem:root-characterization}, we can write
  \begin{equation*}
    \ff(u,v) = \alpha u^{r}v^{r}\prod_{i=1}^{n-r} \left( u-\lambda_{i}v\right) \left( \overline{\lambda_{i}}u+v\right).
  \end{equation*}
  Let
  \begin{equation*}
    \bw_{i} = \left\{
    \begin{array}{ll}
      \left( u-\lambda_{i}v\right) \left( \overline{\lambda_{i}}u+v\right), & \hbox{if $i \le n-r$;} \\
      uv,                                                                   & \hbox{if $i > n-r$.}
    \end{array}
    \right.
  \end{equation*}
  Each $\bw_{i}$ is in $\SRn{2}$, hence $\psi^{-1}(\bw_{i}) \in \Hn{1}(\RR^{3})$ is a polynomial of degree $1$ on $\RR^{3}$ which can be written as $\xx \cdot \ww_{i}$ for some vector $\ww_{i}$ in $\RR^{3}$. We get thus
  \begin{multline*}
    \rp = \psi^{-1}(\ff) = \psi^{-1}(\bw_{1} \dotsb \bw_{n})
    \\
    = \psi^{-1}(\bw_{1}) \ast \dotsb \ast \psi^{-1}(\bw_{n})  = (\xx \cdot \ww_{1}) \ast \dotsb \ast (\xx \cdot \ww_{n}),
  \end{multline*}
  which achieves the proof.
\end{proof}

\begin{proof}[Proof of Proposition~\ref{prop:fiber}]
  Let $\ww_{1},\dotsc,\ww_{n}$ and $\tilde{\ww}_{1},\dotsc,\tilde{\ww}_{n}$ be unit vectors. We can assume, rotating if necessary this configuration of $2n$ vectors on the sphere, that none of the $\ww_{k},\tilde{\ww}_{k}$ are equal to $(0,0,1)$ or $(0,0,-1)$. Set
  \begin{equation*}
    \lambda_{k}:=\tau(\ww_{k}),\quad \mu_{k}:= \tau(\tilde{\ww}_{k}),
  \end{equation*}
  and note that the $\lambda_{k}, \mu_{k}$ belong to $\CC\setminus\set{0}$. Using~\eqref{eq:stereographic-projection-representation}, the relation
  \begin{equation*}
    (\xx \cdot \ww_{1}) \ast \dotsb \ast (\xx \cdot \ww_{n})=(\xx \cdot \tilde{\ww}_{1}) \ast \dotsb \ast (\xx \cdot \tilde{\ww}_{n}), \qquad \forall \xx \in \RR^{3}.
  \end{equation*}
  can be recast as the binary form identity
  \begin{equation*}
    \prod_{k=1}^{n} \alpha_{k}( u +\lambda_{k} v)(\overline{\lambda_{k}}u-v) = \prod_{k=1}^{n} \beta_{k}( u +\mu_{k} v)(\overline{\mu_{k}}u-v), \qquad \forall (u,v) \in \CC^{2}.
  \end{equation*}
  Therefore, there exists a permutation $\sigma$ of $\set{1,\dotsc,n}$, such that either $\mu_{k}=\lambda_{\sigma(k)}$ or $\mu_{k}=-1/\overline{\lambda_{\sigma(k)}}$. In the first case we deduce that
  \begin{equation*}
    \tilde{\ww}_{k} = \tau^{-1}(\lambda_{\sigma(k)}) = \ww_{\sigma(k)}
  \end{equation*}
  while in the second case, we have
  \begin{equation*}
    \tilde{\ww}_{k} = \tau^{-1}\left(-\frac{1}{\overline{\lambda_{\sigma(k)}}}\right )= -\ww_{\sigma(k)}.
  \end{equation*}
  We obtain therefore that $\tilde{\ww}_{k} = \epsilon_{k}\ww_{\sigma(k)}$ where $\epsilon_{k} = \pm 1$ for $1 \le k \le n$ with $\epsilon_{1} \dotsb \epsilon_{n} = 1$. If the vectors $\ww_{k}, \tilde{\ww}_{k}$ are not unit vectors, a normalization process leads to the conclusion that $\tilde{\ww}_{k} = c_{k}\ww_{\sigma(k)}$ where $c_{k} \in \RR$ and $c_{1} \dotsb c_{n} = 1$, which achieves the proof.
\end{proof}

\section{Symmetry classes}
\label{sec:symmetry-classes}

The action of the rotation group $\SO(3)$ on the space of $n$-order harmonic tensors $\HH^{n}(\RR^{3})$ partitions this space into \emph{symmetry classes}, where two tensors $\bH_{1},\bH_{2}$ belong to the same symmetry class if their respective symmetry group $G_{\bH_{1}}$ and $G_{\bH_{2}}$ are conjugates, that is
\begin{equation*}
  G_{\bH_{2}} = gG_{\bH_{1}}g^{-1}, \qquad \text{for some $g \in \SO(3)$.}
\end{equation*}
Symmetry classes are in correspondence with conjugacy classes
\begin{equation*}
  [K] := \set{ gKg^{-1};\; g\in \SO(3)}
\end{equation*}
of closed subgroups $K$ of $\SO(3)$ (several conjugacy classes may however correspond to empty symmetry classes). Moreover, symmetry classes are partially ordered by the following partial order defined on conjugacy classes
\begin{equation*}
  [K_{1}]\preceq [K_{2}] \qquad \text{if} \qquad \exists g\in \SO(3),\quad K_{1}\subset gK_{2}g^{-1}.
\end{equation*}

Recall that a closed $\SO(3)$ subgroup is conjugate to one of the following list \cite{GSS1988}:
\begin{equation*}
  \SO(3),\, \OO(2),\, \SO(2),\, \DD_{n} (n \ge 2),\, \ZZ_{n} (n \ge 2),\, \tetra,\, \octa,\, \ico,\, \text{and}\, \triv
\end{equation*}
where:
\begin{itemize}
  \item $\OO(2)$ is the subgroup generated by all the rotations around the $z$-axis and the order 2 rotation $\sigma : (x,y,z)\mapsto (x,-y,-z)$ around the $x$-axis.
  \item $\SO(2)$ is the subgroup of all the rotations around the $z$-axis.
  \item $\ZZ_{n}$ is the unique cyclic subgroup of order $n$ of $\SO(2)$, the subgroup of rotations around the $z$-axis.
  \item $\DD_{n}$ is the \emph{dihedral} group. It is generated by $\ZZ_{n}$ and $\sigma :(x,y,z)\mapsto (x,-y,-z)$.
  \item $\tetra$ is the \emph{tetrahedral} group, the (orientation-preserving) symmetry group of a tetrahedron. It has order 12.
  \item $\octa$ is the \emph{octahedral} group, the (orientation-preserving) symmetry group of a cube or octahedron. It has order 24.
  \item $\ico$ is the \emph{icosahedral} group, the (orientation-preserving) symmetry group of a icosahedra or dodecahedron. It has order 60.
  \item $\triv$ is the trivial subgroup, containing only the unit element.
\end{itemize}

The symmetry classes of a second order symmetric tensor are described by the following totally ordered set:
\begin{equation*}
  \begin{array}{c}
    \textrm{orthotropic}
    \\
    \textrm{$[\mathbb{D}_{2}]$}
  \end{array}
  \longrightarrow
  \begin{array}{c}
    \textrm{transversely isotropic}
    \\
    \textrm{$[\OO(2)]$}
  \end{array}
  \longrightarrow
  \begin{array}{c}
    \textrm{isotropic}
    \\
    \textrm{$[\SO(3)]$}
  \end{array}
\end{equation*}

For $\HH^{4}(\RR^{3})$, it is known that there are only $8$ symmetry classes (the same classes as for the Elasticity tensor~\cite{FV1996}). The corresponding partial ordering is illustrated in figure~\ref{fig:lattice}, with the convention that a subgroup at the starting point of an arrow is conjugate to a subgroup of the group pointed by the arrow.

\begin{figure}[h!]
  \includegraphics[scale=1]{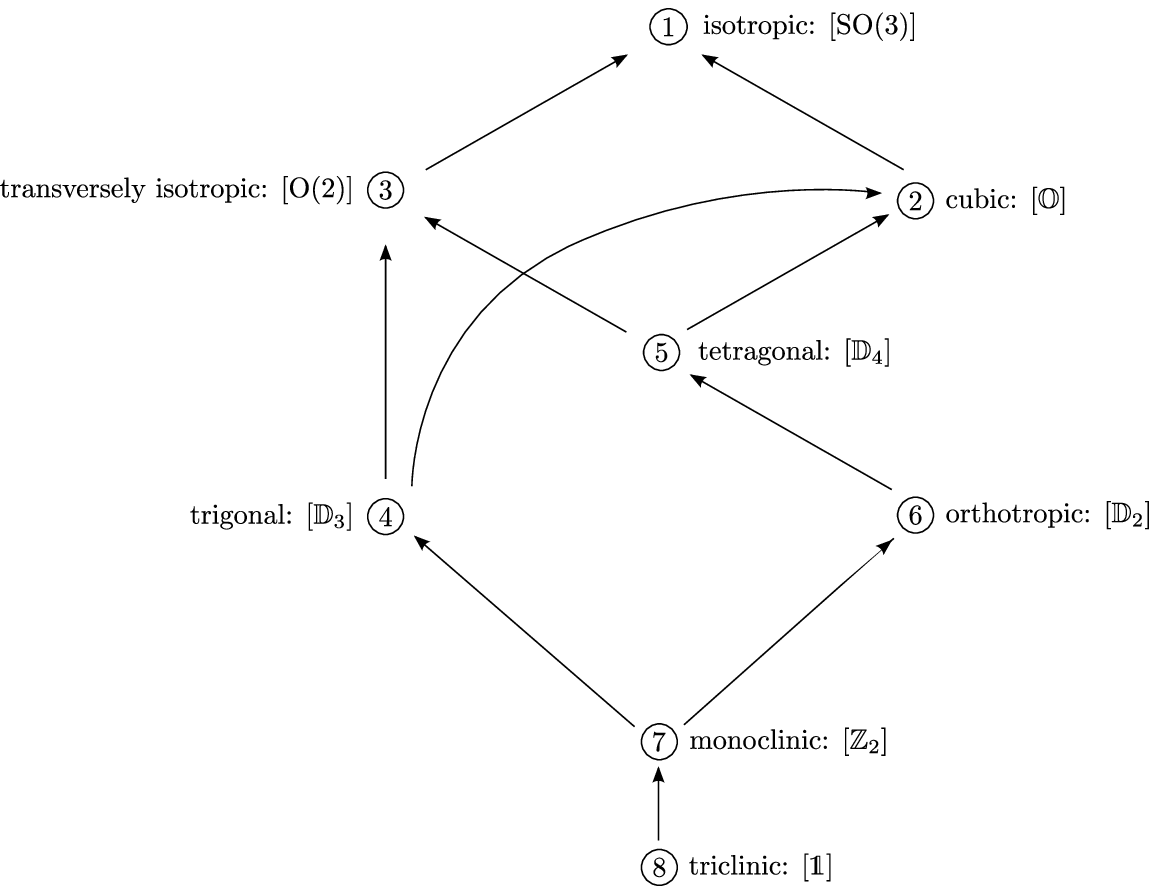}
  \caption{The partial ordering set for symmetry classes of $\HH^{4}(\RR^{3})$.}
  \label{fig:lattice}
\end{figure}


\end{document}